\documentclass{article}
\usepackage{authblk}

\usepackage{booktabs} 
\usepackage[ruled]{algorithm2e} 

\SetAlFnt{\small}
\SetAlCapFnt{\small}
\SetAlCapNameFnt{\small}
\SetAlCapHSkip{0pt}
\IncMargin{-\parindent}

\newcommand{\Nat}{{\mathbb N}}

\usepackage{graphicx}
\usepackage{tikz}
\usetikzlibrary{arrows.meta,shapes,positioning,calc,automata,quotes,patterns}
\tikzset{
  >={To[scale length=1.5]},
  intuitionistic relation/.style={->},
  successor relation/.style={->,dashed},
  model state/.style={circle, inner sep=0pt, minimum size=.3cm},
  p model state/.style={model state, color=gray!30!black, fill, text=white},
  not p model state/.style={model state, draw, fill=white},
  big model state/.style={circle, inner sep=0pt, minimum size=.42cm},
  p big model state/.style={big model state, color=gray!30!black, fill, text=white},
  not p big model state/.style={big model state, draw, fill=white}
}

\usepackage{macros,stmaryrd,hyperref,tabularx,amsmath,amsthm, amssymb}
\usepackage{float}
\usepackage[inline]{enumitem}
\newtheorem{theorem}{Theorem}
\newtheorem{lemma}{Lemma}
\newtheorem{definition}{Definition}
\newtheorem{proposition}{Proposition}
\newtheorem{corollary}{Corollary}
\newtheorem{question}{Question}

\newcommand{\Qfin}{^{a \leftarrow b}}

\newcommand{\toremove}[1]{}
\newcommand{\david}[1]{}
\newcommand{\philippe}[1]{}

\newcommand{\langred}{\lang_{{\ubox}{\Until}}}
\newcommand{\HTMod}[1]{\model^{\ubox}_{#1}}
\newcommand{\DiamMod}[1]{\model^{\diam}_{#1}}

\usepackage{tikz}
\usetikzlibrary{arrows,backgrounds,quotes,calc}

\newcommand{\fullproof}[1]{}
\newcommand{\shortproof}[1]{{\color{orange} #1}}

\makeatletter
\newcommand*{\ineq}[2][]{%
  \begingroup
    \refstepcounter{equation}%
    \ifx\\#1\\%
    \else
      \label{#1}%
    \fi
    \relpenalty=10000 %
    \binoppenalty=10000 %
    \@eqnnum \ \ensuremath{%
      #2%
    }%
  \endgroup
}
\makeatother

\newcommand{\seq}{\succcurlyeq}
\newcommand{\Var}{{\mathbb P}}

\newcommand{\fullmod}[2]{{#1}=(W_{#2},{\peq}_{#2},V_{#2})}
\newcommand{\fulllab}[2]{{#1}=(W_{#2},{\peq}_{#2},\labfun_{#2})}
\newcommand{\length}[1]{\lvert#1\rvert}

\newcommand{\ur}[2]{{\rm ur}_{#1} ( {#2} ) }
\newcommand{\T}{{\mathfrak T}}
\newcommand{\Grph }{{\mathfrak G}}
\newcommand{\M }{{\mathfrak M}}
\newcommand{\N }{{\mathfrak N}}
\newcommand{\iM }{{\mathfrak M}}

\newcommand{\newmodel}[1]{{#1}^{\rm mod}}
\newcommand{\A}{{\mathfrak A}}
\newcommand{\iA}{{\mathfrak A}}
\newcommand{\B }{{\mathfrak B}}
\newcommand{\iB }{{\mathfrak B}}
\newcommand{\dpt}[1]{{\rm dpt}(#1)}
\newcommand{\degp}[1]{{\rm lev}(#1)}
\newcommand{\Labels}{\Lambda}
\newcommand{\simvar}{\sigma}
\newcommand{\srone}{\rho}
\newcommand{\srtwo}{\iota}
\newcommand{\ignore}[1]{}
\newcommand{\peq}{\preccurlyeq}
\newcommand{\labfun}{\lambda}
\newcommand{\sfun}{\msuccfct}
\newcommand{\itlht}{{\sf ITL^{ht}}}
\newcommand{\iltl}{{\sf ITL^e}}
\newcommand{\itlb}{{\sf ITL^p}}
\newcommand{\lang}{{\sf L}}
\newcommand{\rng}{{\rm rng}}
\newcommand{\imm}{\mathrel\unlhd}
\newcommand{\cond}{\ll}
\newcommand{\simm}{\mathrel\triangleq}

\newcommand{\wsf}[1]{\Sigma(#1)}

\newcommand{\Edefect}{{D}} 
\newcommand{\Qbase}[1][]{}  
\newcommand{\Qstrat}[1][]{^{{\str}\ifx#1\relax \relax\else, #1\fi}}
\newcommand{\str}{{\rm e }} 

\newcommand{\bgraph}[2][\Labels]{\Grph ^{#1}_{#2}}
\newcommand{\bgraphbase}[2][\Labels]{\bgraph[#1]{#2} = %
  \left(W^{#1}_{#2}, \isuccrel^{#1}_{#2}, \labfun^{#1}_{#2}\right)}

\usepackage{xypic}
\usepackage{graphicx}
\usepackage{multicol}
\newcommand{\eqdef}{:=}
\usepackage{MnSymbol}
\usepackage{subcaption}
\renewcommand{\tnext}{\medcircle}
\def\R{\mathop{\sf R}}
\def\Until{\mathop{\sf U}}
\def\bisim{\mathrel{ Z}}

\begin{document}
\title{Intuitionistic Linear Temporal Logics}  
\author[1]{Philippe Balbiani. \footnote{\href{mailto:philippe.balbiani@irit.fr}{\tt philippe.balbiani@irit.fr}}}
\author[1]{Joseph Boudou. \footnote{\href{mailto:joseph.boudou@irit.fr}{\tt joseph.boudou@irit.fr}}}
\author[2]{Mart\'in Di\'eguez \footnote{\href{mailto:martin.dieguez@enib.fr}{\tt martin.dieguez@enib.fr}}}
\author[3]{David Fern\'andez-Duque \footnote{\href{mailto:david.fernandezduque@ugent.be}{\tt david.fernandezduque@ugent.be}}}
\affil[1]{IRIT, Toulouse University. Toulouse, France}
\affil[2]{CERV, ENIB,LAB-STICC. Brest, France}
\affil[3]{Department of Mathematics, Ghent University. Gent, Belgium}
\date{~}
\maketitle

\begin{abstract}
We consider intuitionistic variants of linear temporal logic with `next', `until' and `release' based on {\em expanding posets:} partial orders equipped with an order-preserving transition function.
This class of structures gives rise to a logic which we denote $\iltl$, and by imposing additional constraints we obtain the logics $\itlb$ of {\em persistent posets} and $\itlht$ of {\em here-and-there temporal logic,} both of which have been considered in the literature.
We prove that $\iltl$ has the effective finite model property and hence is decidable, while $\itlb$ does not have the finite model property.
We also introduce notions of bounded bisimulations for these logics and use them to show that the `until' and `release' operators are not definable in terms of each other, even over the class of persistent posets.
\end{abstract}


\section{Introduction}

Intuitionistic logic~\cite{DVD86,MInt}
and its modal extensions~\cite{Ewald,PS86,Simpson94}
play a crucial role in computer science and artificial intelligence {and {\em Intuitionistic Temporal Logics} have not been an exception. The study of these logics can be a challenging enterprise~\cite{Simpson94} and, in particular, there is a huge gap that must be filled regarding combinations of intuitionistic and linear-time temporal logic~\cite{pnueli}. This is especially pressing given several potential applications of intuitionistic temporal logics that have been proposed 
by several authors.	

The first involves the Curry-Howard correspondence~\cite{Howard80}, which identifies intuitionistic proofs with the $\lambda$-terms of functional programming. Several extensions of the $\lambda$-calculus with operators from {\em Linear Time Temporal Logic}~\cite{pnueli} ($\sf LTL$) have been proposed in order to introduce new features to functional languages:
Davies \cite{Davies96,Davies-2017} has suggested adding a `next' ($\tnext$) operator to intuitionistic logic in order to define the type system $\lambda^\tnext$, which allows extending functional languages with  \emph{staged computation}\footnote{\emph{Staged computation} is a technique that allows dividing the computation in order to exploit the early availability of some arguments.}~\cite{Ershov77}.
Davies and Pfenning~\cite{Davies01} proposed the functional language ${\sf Mini\text{-}ML}^\Box$ which is supported by intuitionistic $\sf S4$ and allows capturing complex forms of staged computation as well as runtime code generation.
Yuse and Igarashi later extended $\lambda^\tnext$ to $\lambda^\Box$~\cite{Yuse2006} by incorporating the `henceforth' operator ($\ubox$), useful for modelling persistent code that can be executed at any subsequent state.  

Alternately, intuitionistic temporal logics have been proposed as a tool for modelling semantically-given processes.
Maier~\cite{Maier2004Heyting} observed that an intuitionistic temporal logic with `henceforth' and `eventually' ($\diam$) could be used for reasoning about safety and liveness conditions in possibly-terminating reactive systems, and
Fern\'andez-Duque~\cite{DFD2016} has suggested that a logic with `eventually' can be used to provide a decidable framework in which to reason about topological dynamics.
In the areas of {nonmonotonic reasoning,} {knowledge representation (KR),} and {artificial intelligence,} intuitionistic and intermediate logics have played an important role within the successful {answer set programming (ASP)}~\cite{Brewka11} paradigm for practical KR, leading to several extensions of modal ASP~\cite{CP07} that are supported by intuitionistic-based modal logics like \emph{temporal here and there}~\cite{BalbianiDieguezJelia}.}



%

{There have been some notable steps towards understanding intuitionisitic temporal logics:

%
\begin{itemize}\itemsep0pt	
	\item Davies' intuitionistic temporal logic with $\tnext$ \cite{Davies96} was provided Kripke semantics and a complete deductive system by Kojima and Igarashi \cite{KojimaNext}.	
	\item Logics with $\tnext,\ubox$ were axiomatized by Kamide and Wansing \cite{KamideBounded}, where $\ubox$ was interpreted over bounded time.	
  \item Balbiani and Di\'eguez~\cite{BalbianiDieguezJelia} axiomatized the Here and There~\cite{Hey30} variant of $\sf LTL$ with $\tnext,\diam,\ubox$.	
	\item Davoren \cite{Davoren2009} introduced topological semantics for temporal logics and Fern\'andez-Duque \cite{DFD2016} proved the decidability of a logic with $\tnext,\diam$ and a universal modality based on topological semantics.		
\end{itemize}
Nevertheless, many questions have remained open, especially regarding conservative extensions of intuitionistic logic with all of the tenses $\tnext,\diam,\ubox$, or even the more expressive `until' $\Until$ and `release' $\R$.

With the exception of \cite{Davoren2009,DFD2016}, semantics for intuitionistic $\sf LTL$ use frames
of the form $(W,{\irel},{\msuccfct})$, where $\irel$ is a partial order used to interpret the
intuitionistic implication and $\msuccfct$ is a binary relation used to interpret temporal
operators.}
Since we are interested in linear time, we will restrict our attention to the case where $\msuccfct$ is a function. Thus, for example, $\tnext p$ is true at some world $w\in W$ whenever $p$ is true at $\msuccfct(w)$.
Note, however, that $\msuccfct$ cannot be an arbitrary function. Intuitionistic semantics have the feature
that, for any formula $\varphi$ and worlds $w\irel v\in W$, if $\varphi$ is true at $w$ then it must
also be true at $v$; that is, truth is {\em monotone} (with respect to $\irel$). If we want this property to be preserved by
formulas involving $\tnext$, we need for $\peq$ and $\msuccfct$ to satisfy certain confluence
properties. In the literature, one generally considers frames satisfying
\begin{enumerate}
\item \label{ItForCon}  $w\irel v$ implies $\msuccfct(w)\irel\msuccfct(v)$ ({\em forward confluence,} or simply {\em confluence}), and
\item\label{ItBackCon}  if $u\succcurlyeq\msuccfct(w)$, there is $v\succcurlyeq w$ such that $\msuccfct(v)=u$ ({\em backward confluence})
\end{enumerate}
(see Figure \ref{FigCO}). We will call frames satisfying these conditions {\em persistent frames} (see Sec.~\ref{SecExpPer}),
mainly due to the fact that they are closely related to (persistent) products of modal logics
\cite{mdml}.
Persistent frames for intuitionistic $\sf LTL$ are closely related to the frames of the modal logic ${\sf LTL} \times {\sf S4}$, which is non-axiomatizable. For this reason, it may not be surprising that it is unknown whether the intuitionistic temporal logic of persistent frames, which we denote $\itlb$, is decidable.

However, as we will see in Proposition \ref{PropIntCond}, only forward confluence is needed for truth of all formulas to be monotone, even in the presence of $\diam$, $\ubox$ or even $\Until$ and $\R$. The frames satisfying this condition are, instead, related to {\em expanding} products of modal logics \cite{pml}, which are often decidable even when the corresponding product is non-axiomatizable. This suggests that dropping the backwards confluence could also lead to a more manageable intuitionistic temporal logic. We denote the resulting logic by $\iltl$ and, as we will prove in this paper, it enjoys a crucial advantage over $\itlb$: $\iltl$ has the effective finite model property (hence it is decidable), but $\itlb$ does not. In fact, to the best of our knowledge, $\iltl$ is the first known decidable intuitionistic temporal logic that
\begin{enumerate}

	\item is conservative over propositional intuitionistic logic,

	\item includes (or can define) the three tenses $\tnext,\Until,\R$, and

	\item is interpreted over infinite time.

\end{enumerate}
Intuitively, $\itlb$ is a logic of invertible processes, while $\iltl$ reasons about non-invertible ones.
The latter is closely related to $\mathsf{ITL^c}$, an intuitionistic temporal logic for continuous dynamic topological systems \cite{DFD2016}.
In contrast, the logic $\iltl$ is based on relational, rather than topological, semantics, which has the advantage of admitting a natural `henceforth' operator (although topological variants can be defined \cite{BoudouJELIA}).
The current work extends previous results regarding a variant of $\iltl$ with $\diam$ and $\ubox$, rather than $\Until$ and $\R$ \cite{Boudou2017}.

Note that $\diam \varphi \equiv \neg \ubox \neg \varphi$ is not valid intuitionistically and hence $\diam$ cannot be defined in terms of $\ubox$ using the standard equivalence.
The same situation holds for the `until' operator: while the language with $\tnext$ and $\Until$ is equally expressive to classical monadic first-order logic with $\leq$ over $\mathbb N$~\cite{Gabbay80}, $\Until$ admits a first-order definable intuitionistic dual, $\R$ (`release'), which cannot be defined in terms of $\Until$ using the classical definition.

However, this is not enough to conclude that $\R$ cannot be defined in a different way in terms of $\Until$. Thus we will consider the question of {\em definability:} which of the modal operators can be defined in terms of the others? As is well-known, $\diam \varphi \equiv \top \Until \varphi$ and $\ubox \varphi \equiv \bot \R \varphi$; these equivalences remain valid in the intuitionistic setting. Nevertheless, we will show that $\ubox$ cannot be defined in terms of $\Until$, and $\diam$ cannot be defined in terms of $\R$; in order to prove this, we will develop a theory of bisimulations on $\iltl$ models.

\subsection*{Layout}
{The paper is organised as follows: 
	in Section~\ref{SecSynSem} we present the syntax and the semantics in terms of dynamic posets and also study the validity of some of the classical axioms in our setting. In Section~\ref{SecExpPer} we present the concepts of stratified and expanding frames and also show that satisfiability and validity on arbitrary models is equivalent to satisfiability and validity on expanding models. In Section~\ref{secSpecial} we consider two smaller classes of models, persistent and here-and-there models, and we compare their logics to $\iltl$.
	
	The Finite Model Property of $\iltl$ is studied along sections~\ref{sec:combinatorics} and~\ref{SecDec}. In the former we introduce the concepts of labelled structures and quasimodels as well as several related concepts such as immersions, condensations, and normalised quasimodels. Those definitions are used in Section~\ref{SecDec} to prove the finite model property of $\iltl$.
	
In Section~\ref{sec:bisim} we define the concept of bounded bisimulations in intuitionistic modal setting and use them to study the interdefinability of the $\iltl$ modalities in Section~\ref{SecSucc}. We finish the paper with conclusions and future work.
	
}


\section{Syntax and semantics} \label{SecSynSem}

We will work in sublanguages of the language $\lang$ given by the following grammar:
\begin{equation*}
  \varphi, \psi \gramdef p \alt \bot \alt
    ( \varphi \wedge \psi ) \alt ( \varphi \vee \psi ) \alt (\varphi \imp \psi) \alt
    ( \tnext \varphi ) \alt ( \diam\varphi ) \alt ( \ubox\varphi) \alt ( \varphi \Until \psi ) \alt ( \varphi \R \psi)
\end{equation*}
where $p$ is an element of a countable set of propositional variables $\Var$. Henceforth we adhere to the standard conventions for omission of parentheses. All sublanguages we will consider include all Boolean operators and $\tnext$, hence we denote them by displaying the additional connectives as a subscript: for example, $\lang_{\diam\ubox}$ denotes the $\Until$-free, $\R$-free fragment. As an exception to this general convention, $\lang_{\tnext}$ denotes the fragment without $\diam,\ubox,{\Until}$ or $\R$.

Given any formula $\varphi $,
we define the \emph{length} of $\varphi$ (in symbols, $\length{\varphi}$) recursively as follows:

\begin{itemize}[itemsep=0pt]
	\item $\length{p} = \length{\bot} = 0 $;
	\item $\length{\phi \odot \psi}= 1 + \length{\phi} + \length{\psi}$, with $\odot \in \lbrace \vee, \wedge,\imp, \R, \Until \rbrace$;	
	\item $\length{\odot \psi}= 1 + \length{\psi}$, with $\odot \in \lbrace \neg, \tnext, \ubox, \diam \rbrace$.
\end{itemize}
\noindent Broadly speaking, the length of a formula $\varphi$ corresponds to the number of connectives appearing in $\varphi$.

\subsection{Dynamic posets}

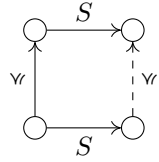
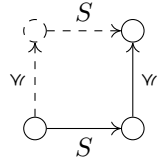
\begin{figure}
\begin{center}

\begin{subfigure}{.5\columnwidth}\centering
  \begin{tikzpicture}[x=1.3cm,y=1.3cm]
    \foreach \n/\x/\y in {bl/0/0, br/1/0, tl/0/1, tr/1/1}
      \node (\n) [not p model state] at (\x,\y) {};

    \path[solid]
      (bl) edge[->] node [below] {$S$} (br)
      (tl) edge[->] node [above] {$S$} (tr)
      (bl) edge[->] node [above,sloped] {$\peq$} (tl)
    ;
    \path[dashed]
      (br) edge[->] node [below,sloped] {$\peq$} (tr)
    ;
  \end{tikzpicture}
  \caption{Forward confluence}
\end{subfigure}
~
\begin{subfigure}{.5\columnwidth}\centering
  \begin{tikzpicture}[x=1.3cm,y=1.3cm]
    \foreach \n/\x/\y in {bl/0/0, br/1/0, tr/1/1}
      \node (\n) [not p model state] at (\x,\y) {};
    \node (tl) [not p model state, dashed] at (0,1) {};

    \path[solid]
      (bl) edge[->] node [below] {$S$} (br)
      (br) edge[->] node [below,sloped] {$\peq$} (tr)
    ;
    \path[dashed]
      (tl) edge[->] node [above] {$S$} (tr)
      (bl) edge[->] node [above,sloped] {$\peq$} (tl)
    ;
  \end{tikzpicture}
  \caption{Backward confluence}
\end{subfigure}

\end{center}

\caption{On a dynamic poset the above diagrams can always be completed if $S$ is forward or backward confluent, respectively. Posets with both properties are {\em persistent.}}\label{FigCO}
\end{figure}

Formulas of $\lang$ are interpreted over dynamic posets. A {\em dynamic poset} is a tuple $ D = (W,\peq,\sfun)$,
where 
$W$ is a non-empty set of states,
$\irel$ is a partial order,
and
$\sfun$ is a function from $W$ to $W$ satisfying the {\em forward confluence} condition that for all $ w, v \in W, $ if $ w \irel v $ then $ \msuccfct(w) \irel \msuccfct(v).$ An {\em intuitionistic dynamic model,} or simply {\em model,} is a tuple $\modelbase$ consisting of a dynamic poset equipped
with a valuation function $\deffun V   W {\parts\Var} $ that is {\em monotone} in the sense that
for all $ w, v \in W, $ if $ w \irel v $ then $ V(w) \subseteq V(v).$ In the standard way, we define $\msuccfct[^0](w) = w$ and,
for all $k \geq 0$, $\msuccfct[^{k+1}](w) = \msuccfct\left(\msuccfct[^{k}](w)\right)$. Then we define the satisfaction relation $\sat$ inductively by:

\noindent\begin{enumerate}[itemsep=0pt]
	\item $\model, w \sat p $  iff $p \in V(w) $;
	\item $\model, w \nsat \bot$;
	\item $\model, w \sat \varphi\wedge \psi$  iff $  \model, w \sat \varphi $ and $\model, w \sat \psi$;
	\item $\model, w \sat \varphi \vee \psi $  iff  $ \model, w \sat \varphi $ or $\model, w \sat \psi$;
	\item $\model, w \sat \tnext \varphi $   iff $ \model, \msuccfct(w) \sat \varphi $;
		\item  $\model, w \sat \varphi \imp \psi $ iff $\forall v \succcurlyeq w $, if $\model, v \sat \varphi$, then $\model, v \sat \psi $;
	\item $\model, w \sat \diam \varphi $  iff there exists $ k \ge 0$ such that $ \model, \msuccfct[^k](w) \sat \varphi$;
	\item $\model, w \sat \ubox \varphi $  iff for all $ k \ge 0 $ we have that $ \model, \msuccfct[^k](w) \sat \varphi$;
	\item  $\model, w \sat \varphi \Until \psi $ iff there exists $k \ge 0 $ such that $ \model, \msuccfct[^k](w) \sat \psi$ and $\forall i \in [0,k)$, $\model, \msuccfct[^i](w) \sat \varphi$;	
	\item 	$\model, w \sat \varphi \R \psi $ iff for all $k \ge 0$, either $\model, \msuccfct[^k](w) \sat \psi$ or $ \exists i \in [ 0 , k)$ such that $\model, \msuccfct[^i](w) \sat \varphi$.
\end{enumerate}
See Figure \ref{fig:iltl-frame} for illustration of the `$\sat$' relation.
Given a model $\modelbase$
and $w\in W$,
we write $\Sigma_\model(w)$ for the set $\ens{\psi \in \Sigma}{\model, w \sat \psi}$; the subscript `$\model$' is omitted when it is clear from the context.

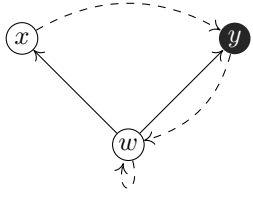
\begin{figure}\centering
	\begin{tikzpicture}[node distance=2cm] 
  \node (a) [not p big model state] {$w$};
  \node (b) [not p big model state, above left of=a] {$x$};
  \node (c) [    p big model state, above right of=a] {$y$};
	
	\path[intuitionistic relation]
  (a) edge[] node {} (b)
	(a) edge[] node {} (c)
  ;
  \path[successor relation]
	(c) edge[bend left] node {} (a)
	(b) edge[bend left] node {} (c)
	(a) edge[loop below] node {} (a)
	;					
	\end{tikzpicture}				
	\caption{Example of an $\iltl$ model $ \model = (W,{\peq},S,V)$, where
    $\peq$ is the reflexive and transitive closure of the relation indicated by the solid arrows,
    $S$ is the relation indicated by the dashed arrows, and
    a black dot indicates that the variable $p$ is true, so that we only have $p\in V(y)$.
    Then, the reader may verify that $\model, x\sat \tnext p$ but $\model, x \nsat p$, while $\model, y \sat p$ but $\model, y \nsat \tnext p$.
    From this it follows that $\model, w \nsat (\tnext p\to p)\vee ( p \to \tnext p)$.}
	\label{fig:iltl-frame}
\end{figure}	

A formula $\varphi$ is {\em satisfiable over a class $\Omega$ of models} if
there is a model $\model\in \Omega$ and a world~$w$ so that $\model,w\sat\varphi$, and {\em valid over $\Omega$} if, for every world $w$ of every model $\model\in \Omega$ we have that $\model,w\sat\varphi$.
Satisfiability (resp. validity) over the class of all intuitionisitic dynamic models
is called {\em satisfiability} (resp. \emph{validity}) for
the {\em expanding domain intuitionisitic temporal logic}~$\iltl$.
We will justify this terminology in the next section. First, we remark that dynamic posets impose the minimal conditions on $\msuccfct$ and $\irel$ in order to preserve the monotonicity of truth of formulas, in the sense that if $\model, w \sat \varphi $ and $w\irel v$ then $\model, v \sat \varphi $. Below, we will use the notation $\llbracket\varphi\rrbracket=\{w\in W \mid \model,w\sat \varphi \}$.

\begin{proposition}\label{PropIntCond}
Let $\mathfrak D=( W,{\irel},\msuccfct)$, where $(W,{\irel})$ is a poset and $\msuccfct\colon W\to W$ is any function. Then, the following are equivalent:
\begin{enumerate}

\item\label{ItExpOne} $S$ is forward confluent;

\item\label{ItExpTwo} for every valuation $V$ on $\mathfrak D$ and every formula $\varphi$, truth of $\varphi$ is monotone with respect to $\irel$.

\end{enumerate}
\end{proposition}

\begin{proof}
That \eqref{ItExpOne} implies \eqref{ItExpTwo} follows by a standard structural induction on $\varphi$. The case where $\varphi\in \mathbb P$ follows from the condition on $V$ and most inductive steps are routine. Consider the case where $\varphi= \psi \Until \theta$, and suppose that $w\irel v$ and $w\in\llbracket \varphi \rrbracket$. Then there exists $k \in \mathbb N$ such that $\model, \msuccfct[^k](w) \sat \theta$ and for all $i\in [0,k)$, $\model, \msuccfct[^i](w) \sat \psi$. Since $S$ is confluent, an easy induction shows that, for all $i\in [0,k]$, $\msuccfct[^i](w) \irel \msuccfct[^i](v)$. Therefore, from the induction hypothesis we obtain that $\model, \msuccfct[^k](v) \sat \theta$ and for all $i\in [0,k)$, $\model, \msuccfct[^i](v) \sat \psi$. Other cases are either similar or easier.

Now we prove that \eqref{ItExpTwo} implies \eqref{ItExpOne} by contrapositive. Suppose that $(W,{\irel},{\msuccfct})$ is not forward-confluent, so that there are $w\irel v$ such that $\msuccfct(w)\not\irel \msuccfct (v)$.
Choose $p\in \mathbb P$ and define $V(u)=\{p\}$ if $\msuccfct(w)\irel u$, $V(u)=\varnothing$ otherwise.
It follows from the transitivity of $\peq$ that $V$ is monotone.
However, $p\not\in V(\msuccfct(v))$, from which it follows that $( D,V),w\sat \tnext p$ but $( D,V),v\nsat \tnext p$.
\end{proof}

Observe that satisfiability in propositional intuitionistic logic is equivalent to satisfiability in classical propositional logic. This is because, if $\varphi$ is classically satisfiable, it is trivially intuitionistically satisfiable in a one-world model; conversely, if $\varphi$ is intuitionistically satisfiable, it is satisfiable in a finite model, hence in a maximal world of that finite model, and the generated submodel of a maximal world is a classical model. Thus it may be surprising that the same is not the case for intuitionistic temporal logic:

\begin{proposition}
Any formula $\varphi$ of the temporal language that is classically satisfiable is satisfiable in a dynamic poset. However, there is a formula satisfiable on a dynamic poset that is not classically satisfiable.
\end{proposition}

\proof
If $\varphi$ is satisfied on a classical {\sf LTL} model $\model$, then we may regard $\model$ as an intuitionistic model by letting $\peq$ be the identity.
On the other hand, consider the formula $\neg\tnext p\wedge \neg\tnext\neg p$ (recall that $\neg \theta$ is a shorthand for $\theta\to \bot$).
Classically, this formula is equivalent to $\neg\tnext p\wedge \tnext p$, and hence unsatisfiable.
Define a model $\modelbase$, where $W=\{w,v,u\}$, $x\peq  y$ if $x=y$ or $x=v$, $y=u$, $\sfun (w)=v$
and $\sfun (x)=x$ otherwise, $V(u)=\{ p \}$ and $V(v) = V(w) = \varnothing$ (see Figure \ref{FigIMLA}).
Then, one can check that $\model, w \sat \neg\tnext p\wedge \neg\tnext\neg p$.
\endproof

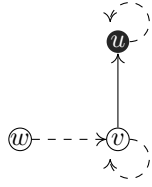
\begin{figure}
\begin{center}
\begin{tikzpicture}[x=1.3cm,y=1.3cm]

  \node (l) [not p model state] at (0,0) {$w$};
  \node (r) [not p model state] at (1,0) {$v$};
  \node (u) [    p model state] at (1,1) {$u$};

  \draw[successor relation] (l) edge (r);
  \draw[successor relation] (r.east) arc (90:-180:.2);
  \draw[successor relation] (u.east) arc (-90:180:.2);
  \draw[intuitionistic relation] (r) edge (u);

\end{tikzpicture}

\end{center}

\caption{A dynamic intuitionistic model. As in the previous figure,
  solid arrows represent the intuitionistic order $\peq$, dashed arrows the successor relation $S$,
  the black point satisfies the atom $p$ and no point satisfies any other atom.
  Note that $S$ is forward, but not backward, confluent.
The world $w$ satisfies $\neg\tnext p\wedge \neg\tnext\neg p$.
  }
\label{FigIMLA}
\end{figure}

Hence the decidability of the intuitionistic satisfiability problem is not a corollary of the
classical case.
In Section~\ref{SecDec}, we will prove that both the satisfiability and the validity problems are
decidable.
We will prove this by showing that $\iltl$ has the effective finite model property: recall that a logic $\Lambda$ has the {\em effective finite model property} for a class of models $\Omega$ if there is a computable function $f\colon \mathbb N \to \mathbb N$ such that given a formula $\varphi$, we have that $\varphi$ is satisfiable (falsifiable) on $\Omega$ if and only if there is $\model \in \Omega$ such that $\varphi$ is satisfied (falsified) on $\model$ and whose domain has at most $f(|\varphi|)$ elements.

\subsection{Some valid and non-valid $\iltl$ formulas}


In this section we present some examples of valid formulas that will be useful throughout the text. We begin by focusing on formulas of $\lang_{\diam\ubox}$.

\begin{proposition}\label{prop:valid:iltl}
The following formulas are $\iltl$-valid:
\begin{multicols}{2}
\begin{enumerate}[itemsep=0pt]
\item $\tnext\bot\leftrightarrow \bot$
\item $\tnext \left( \varphi \wedge \psi \right) \leftrightarrow \left(\tnext \varphi \wedge\tnext \psi\right)$
\item $\tnext \left( \varphi \vee \psi \right) \leftrightarrow \left(\tnext \varphi \vee\tnext \psi\right)$
\item\label{ItValidImp} $\tnext\left( \varphi \rightarrow \psi \right) \rightarrow \left(\tnext\varphi \rightarrow \tnext\psi\right)$
\item $\tnext \ubox \varphi \leftrightarrow \ubox \tnext \varphi$
\item $\tnext \diam \varphi \leftrightarrow \diam \tnext \varphi$

\end{enumerate}
\end{multicols}
\end{proposition}

\begin{proof}
We prove that \ref{ItValidImp} holds and leave other items to the reader. Let $\modelbase$ be any dynamic model and $w \in W$ be such that $\model, w \models \tnext\left( \varphi \rightarrow \psi \right)$. Let $v \seq w$ be such that $\model, v \models \tnext \varphi$. Then, $\model, S(v) \models \varphi$. But $S(w) \peq S(v)$ and $\model, S(w) \models \varphi \to \psi$, so that $\model, S(v) \models \psi$ and $\model, v \models \tnext \psi$. Since $v\seq w$ was arbitrary, $\model, w \models \tnext \varphi \to \tnext \psi$.
\end{proof}




Note that, unlike the other items, \ref{ItValidImp} is not a biconditional, and indeed the converse is not valid over the class of all dynamic posets (see Proposition \ref{prop:nvalid:iltl}).
Next we show that $\diam \varphi$ (resp. $\ubox \varphi$) can be defined in terms of $\Until$ (resp. $\R$) and the $\sf LTL$ axioms involving $\Until$ and $\R$ are also valid in our setting:

\begin{proposition}\label{PropUValid}
The following formulas are $\iltl$-valid:
\begin{multicols}{2}
\begin{enumerate}[itemsep=0pt]
	\item\label{ItUOne} \mbox{$(\varphi \Until \psi) \leftrightarrow \psi \vee \left( \varphi \wedge  \tnext \left(\varphi \Until \psi \right)\right)$}
	\item\label{ItUTwo} \mbox{$(\varphi \R \psi) \leftrightarrow \psi \wedge \left( \varphi \vee  \tnext \left(\varphi \R \psi \right)\right)$}
	\item\label{ItUThree} \mbox{$(\varphi \Until \psi) \imp \diam \psi$}
	\item\label{ItUFour} \mbox{$\ubox \psi \imp (\varphi \R \psi )$}
	\item\label{ItUFive} \mbox{$\diam \varphi \leftrightarrow (\top \Until \varphi)$}
	\item\label{ItUSix} \mbox{$\ubox \varphi \leftrightarrow (\bot \R \varphi)$}
\item 	\label{ItUSeven} \mbox{$\tnext(\varphi \Until \psi)\leftrightarrow ({\tnext\varphi}) \Until ({\tnext\psi})$}
 \item 	\label{ItUEight} \mbox{$\tnext(\varphi \R \psi)\leftrightarrow ({\tnext\varphi})\R({\tnext\psi})$}
 
 \item\label{ItUNine} $\varphi \Until \psi \leftrightarrow (\psi \R (\varphi \vee \psi ))\wedge \diam \psi$
 
  \item\label{ItUTen} $\varphi \R \psi \leftrightarrow ( \psi \Until ( \varphi \wedge \psi ))\vee \ubox \psi $
	
\end{enumerate}
\end{multicols}
\end{proposition}
\begin{proof}
We consider some cases below. For \eqref{ItUOne}, from left to right, let us assume that $\model, w \sat \varphi \Until \psi$. Therefore there exists $k \ge 0$ s.t. $\model, \msuccfct[^k](w) \sat \psi$ and for all $j$ satisfying $0 \le j < k$, $\model, \msuccfct[^j](w) \sat \varphi$. If $k = 0$  then $\model, w\sat \psi$ while, if $k > 0$ it follows that $\model, w \sat \varphi$ and $\model, \msuccfct(w) \sat \varphi \Until \psi$. Therefore $\model, w \sat \psi \vee \left( \varphi \wedge \tnext \varphi\Until\psi\right)$. From right to left, if $\model, w \sat \psi$ then $\model, w \sat \varphi \Until \psi$ by definition (with $k=0$). If $\model, w \sat \varphi \wedge \tnext \varphi \Until \psi$ then $\model, w \sat \varphi$ and $\model,\msuccfct(w)\sat \varphi \Until \psi$ so, due to the semantics, we conclude that $\model, w \sat \varphi \Until \psi$ (with some $k\geq 1$). In any case, $\model, w \sat \varphi \Until \psi$.		

For \eqref{ItUTwo}, we work by contrapositive. From right to left, let us assume that $\model, w \not \sat \varphi \R \psi$. Therefore there exists $k \ge 0$ s.t. $\model, \msuccfct[^k](w) \not\sat \psi$ and for all $j$ satisfying $0 \le j < k$, $\model, \msuccfct[^j](w) \not\sat \varphi$. If $k = 0$  then $\model, w\not \sat \psi$ while, if $k > 0$ it follows that $\model, w \not \sat \varphi$ and $\model, \msuccfct(w) \not \sat \varphi \R \psi$. In any case, $\model, w \not \sat \psi \wedge \left( \varphi \vee \tnext \varphi\R\psi\right)$. From left to right, if $\model, w \not \sat \psi$ then $\model, w \not \sat \varphi \R \psi$ by definition. If $\model, w \not \sat \varphi \vee \tnext \varphi \R \psi$ then $\model, w \not \sat \varphi$ and $\model,\msuccfct(w)\not \sat \varphi \Until \psi$ so, due to the semantics of $\R$, we conclude that $\model, w \not \sat \varphi \R \psi$. In any case, $\model, w \not \sat \varphi \R \psi$.

		
								


The remaining items are left to the reader.
\end{proof}

With these equivalences in mind, we can simplify the syntax of the full language $\lang$.

\begin{proposition}\label{propExpEquiv}
The languages $\lang_{{\diam}{ \R}}$ and $\lang_{{\ubox}{ \Until }}$ are expressively equivalent to $\lang$ over the class of dynamic posets.
\end{proposition}

\begin{proof}
From the validities $\ubox \varphi \leftrightarrow \bot \R \varphi$ and $\varphi \Until \psi \leftrightarrow (\psi \R (\varphi \vee \psi ))\wedge \diam \psi $ we see that any $\varphi \in \lang$ is equivalent to some $\varphi ' \in \lang_{{\diam }{\R}}$. Similarly, from $\diam\varphi \leftrightarrow \top \Until \varphi$ and $\varphi \R \psi \leftrightarrow ( \psi \Until ( \varphi \wedge \psi ))\vee \ubox \psi $ we see that $\lang_{{\ubox }{\Until}}$ is expressively equivalent to $\lang$.
\end{proof}

Nevertheless, we will later show that both $\lang_{\Until}$ and $\lang_{\R}$ are strictly less expressive than the full language, in contrast to the classical case.

\section{The expanding model property} \label{SecExpPer}

As mentioned in the introduction, the logic $\iltl$ is closely related to {\em expanding products} of modal logics \cite{pml}.
In this subsection,
we introduce stratified and expanding frames,
and show that satisfiability and validity on arbitrary models is equivalent to satisfiability and validity on expanding models.
To do this,
it is convenient to represent posets using acyclic graphs.

\begin{definition}
A \emph{directed acyclic graph} is a tuple $(W, \isuccrel)$, where
$W$ is a set of vertices and
${\isuccrel} \subseteq W \times W$ is a set of edges
whose reflexive, transitive closure $\isuccrel^*$ is antisymmetric. We will tacitly identify $(W, \isuccrel)$ with the poset $(W, \isuccrel^*)$. A \emph{path} from $w_1$ to $w_2$
is a finite sequence $v_0 \ldots v_n \in W$ such that $v_0 = w_1$, $v_n = w_2$
and for all $k < n$, $v_k \isuccrel v_{k+1}$. 
A \emph{tree} is an acyclic graph $(W, \isuccrel )$
with an element $r \in W$, called the root, such that for all $w \in W$
there is a unique path from $r$ to $w$.
A poset $(W,\peq)$ is also a tree if there is a relation $\isuccrel$ on $W\times W$ such that
$(W,\isuccrel)$ is a tree and ${\peq}={\isuccrel^*}$.
\end{definition}

Below, if $R\subseteq A\times A$ is a binary relation and $X \subseteq A$, $\reduc{R}{X}$ denotes the restriction of $R$ to $X$. Similarly if $f \colon A\longrightarrow B$ then $\reduc{f}{X}$ denotes the restriction of $f$ to the domain $X$.

\begin{definition} \label{def:stratified}
  A model $\modelbase$ is \emph{stratified} if there is a partition
  $\left\{W_n\right\}_{n < \omega}$ of $W$ such that
  \begin{enumerate}
    \item each $W_n$ is closed under $\irel$,
          \label{cond:stratified:closed}
    \item for all $n$, $(W_n, \reduc{\irel}{W_n})$ is a tree, and
          \label{cond:stratified:tree}
    \item if $w \in W_n$ then $\msuccfct(w) \in W_{n+1}$.
          \label{cond:stratified:succ}
  \end{enumerate}
  If $\model$ is stratified, we write $\peq_n,\msuccfct_n,$ and $V_n$ instead of $\reduc{\irel}{W_n},\reduc{\msuccfct}{W_n},$ and $\reduc{V}{W_n}$.
  We then define $\model_n=(W_n,\irel_n,V_n)$. If moreover we have that $\msuccfct(w)\peq \msuccfct(v)$ implies $w\peq v$, then we say that $\model$ is an {\em expanding model.} We define stratified and expanding posets similarly, ignoring the clauses for $V$.
\end{definition}

Below if $\Sigma \subseteq \Delta\subseteq \lang$ we write $\Sigma \Subset \Delta$ to indicate that $\Sigma$ is finite and closed under subformulas.
In view of Proposition \ref{propExpEquiv}, in this section we may restrict our attention to $\langred$.
Given $\Sigma \Subset \langred$,
a model $\modelbase[\Qbase]$, and a state $w\Qbase \in W\Qbase$,
we will construct a stratified model $\modelbase[\Qstrat]$
such that 
for the root $w\Qstrat$ of $W\Qstrat_{0}$, $\Sigma({w\Qstrat}) = \Sigma({w\Qbase})$.

\begin{definition}\label{defQstrat}
Let $\Sigma \Subset \langred$ and $\modelbase$ be a model.
We first define the set $\Edefect = \nat \times \nat \times \parts\Sigma$ of possible {\em defects,} and fix
an enumeration $  ( (x_k,y_k,H_k)   )_{k \in \nat}$ of $\Edefect$; since $\Sigma$ is finite and not empty, we assume that $\Edefect$ is enumerated such that for each $k > 0$, $x_k \leq k$.
 Then, for each $k \in \nat$, we construct inductively a tuple $(U_k, \isuccrel_k, h_k)$
where $U_k \subseteq \nat \times \nat$, ${\isuccrel_k} \subseteq U_k \times U_k$ and
$\deffun{h_k}{U_k}{W\Qbase}$.
The model $\model\Qstrat$ is defined from these tuples
and the whole construction proceeds as follows:

\paragraph*{Base case.} Let $U_0 = \{0\} \times \nat$, ${\isuccrel_0} = \varnothing$ and 
$h_0$ be such that for all $(0,y) \in U_0$, $h_0(0,y) = {\msuccfct[\Qbase]}^y\left(w\Qbase\right)$.

\paragraph*{Inductive case.} Let $k \geq 0$ and suppose that $(U_k, \isuccrel_k, h_k)$ has already been
constructed.
Let $(x, y, H) = (x_k,y_k,H_k)$.
If
\begin{enumerate*}[label=(D\arabic*)]
  \item\label{cond:stratified:ind:1} $ (x, y) \in U_k$ and \label{cond:stratified:ind:belong}
  \item there is $ v = v_k \in W\Qbase$ such that $h_k(x,y) \irel\Qbase v$ and $\wsf{v} = H$,
    \label{cond:stratified:ind:exist}
\end{enumerate*}
then we construct $(U_{k+1}, \isuccrel_{k+1}, h_{k+1})$ such that:
\begin{align*}
  U_{k+1} &= U_k \cup \ens{( k + 1 , a ) }{y \leq a \in \mathbb N} \\
  \isuccrel_{k+1} &= \mathord{\isuccrel_k} \cup
    \ens{((x,a),(k+1,a)) }{y \leq a \in \mathbb N } \\
    h_{k+1} &= h_k \cup \ens{(( k + 1, a), {\msuccfct[\Qbase]}^{d-y}(v) )}{y \leq  d \in \mathbb N }
\end{align*}
Otherwise $(U_{k+1}, \isuccrel_{k+1}, h_{k+1}) = (U_k, \isuccrel_k, h_k)$.

\paragraph*{Final step.} Let $h =  \bigcup_{k \in \nat} h_k$. We construct $\modelbase[\Qstrat]$ such that $W\Qstrat  = \bigcup_{k \in \nat} U_k$, ${\irel\Qstrat} = \mathord{(\isuccrel\Qstrat)^*},$ where $ \mathord{\isuccrel\Qstrat} = \bigcup_{k \in \nat} \mathord{\isuccrel_k},$ $\msuccfct\Qstrat(a,b) = (a,b+1),$
and $  V\Qstrat(x,y) = V\Qbase\left(h (x,y)\right).
$
\end{definition}

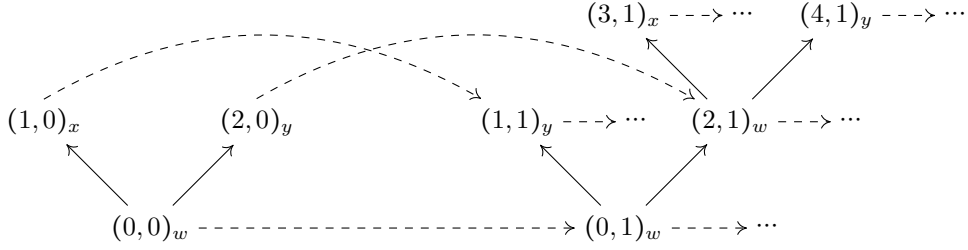
\begin{figure}\centering
	\begin{tikzpicture}[node distance=2cm] 

	\node (a)[] {$(0,0)_w$};
	\node (b)[above left of=a] {$(1,0)_x$};
	\node (c)[above right of=a] {$(2,0)_y$};	
	\node (a10)[right =5cm of a] {$(0,1)_w$};
	\node (c11)[above left of=a10] {$(1,1)_y$};
	\node (a12)[above right of=a10] {$(2,1)_w$};
	\node (b13)[above left of=a12] {$(3,1)_x$};
	\node (c14)[above right of=a12] {$(4,1)_y$};
	\node (sink)[right=.7cm of a12] {$\cdots$};
	\node (sink2)[right=.7cm of b13] {$\cdots$};
	\node (sink3)[right=.7cm of c14] {$\cdots$};
	\node (sink4)[right=1cm of a10] {$\cdots$};
	\node (sink5)[right=.7cm of c11] {$\cdots$};

	\path[intuitionistic relation]
  (a) edge[] node {} (b)
	(a) edge[] node {} (c)
	(a10) edge[] node {} (c11)
	(a10) edge[] node {} (a12)
	(a12) edge[] node {} (b13)
	(a12) edge[] node {} (c14)
  ;
  \path[successor relation]
	(a.east) edge node {} (a10)
	(b.north) edge[bend left] node {} (c11)
	(c.north) edge[bend left] node {} (a12)
	(a12) edge node {} (sink)
	(b13) edge node {} (sink2)
	(c14) edge node {} (sink3)
	(a10) edge node {} (sink4)
	(c11) edge node {} (sink5)
	;				

	\end{tikzpicture}				
	\caption{The strata $W\Qstrat_0$, $W\Qstrat_1$ of the stratified model obtained from the model defined in Figure \ref{fig:iltl-frame}.
    The subindices indicate the value of $h=\bigcup_{k \in \mathbb N} h_k$.}
	\label{fig:expanding-frame}
\end{figure}

See Figure \ref{fig:expanding-frame} for an illustration of the construction.
We wish to prove that the structure $\model\Qstrat $ is a stratified model.
To do this, we first establish some basic properties of the finite stages of the construction.
We begin with some simple observations.

\begin{lemma}\label{lemUStratifiedA}
If $\Sigma \Subset \langred$, $\modelbase$ is any model, $k\in \mathbb N$, and $(U_k)_{k\in \mathbb N}$ is as in Definition \ref{defQstrat}, then
\begin{enumerate}

\item\label{itABound} $(a,b) \in U_k$ implies that $a \leq k$,

\item\label{itSClosed} if $n\in \mathbb N$ then $(S\Qstrat)^n(a,b) = (a,b + n) \in U_k$, and

\item\label{itHHomo} $h_k \colon U_k \to W$ is a function and satisfies $h_k \circ S\Qstrat = S  \circ h_k$.

\end{enumerate}
\end{lemma}

\proof
These claims are proven by a straightforward induction on $k$.
Assume that all claims hold for $i<k$.
If $(a,b) \in U_{k }$ then either $a=k$, or $k>0$ and $(a,b) \in U_{k-1}$. In the former case we trivially have $a=k\leq k$ and in the latter $a\leq k-1$ by the induction hypothesis, establishing \eqref{itABound}.
For \eqref{itSClosed}, if $(a,b) \in U_{k-1}$ then the claim follows easily from the induction hypothesis.
Otherwise, $a=k$. Then, from $y\leq b \leq b + n'$ we see that $(a,b + n') \in U_k$ for all $n'$, so that from the definition of $S\Qstrat $ we obtain $(S\Qstrat)^n(a,b) = (a,b+n) \in U_k$.

Meanwhile $h_{k} (a,b)$ is uniquely defined by either $h_{k} (a,b) = S^{b-y}(v)$ if $a = k$, or $h_{k-1} (a,b) = h_{k} (a,b)$ if $k>0$ and $(a,b) \in U_{k-1}$ (so that $a < k$). From this we see that $h_k(S\Qstrat (a,b)) = h_k(a,b+1) = S^{b+1-y}(v) = S (S^{b-y}(v)) = S (h_k (a,b))$, obtaining \eqref{itHHomo}. 
\endproof

With this, we establish some properties of $\isuccrel\Qstrat_k$.

\begin{lemma}\label{lemUStratifiedB}
Let $\Sigma \Subset \langred$, $\modelbase$ be any model, $k\in \mathbb N$, and $(U_k)_{k\in \mathbb N}$ be defined as in Definition \ref{defQstrat}. Suppose that $(a,b) \isuccrel_k (c,d)$. Then,
\begin{enumerate}

\item\label{itInUk} $(a,b) ,(c,d) \in U_k$,

\item\label{itLessEq} $a < c$ and $b=d$,

\item\label{itInj} if $(a',b') \isuccrel_k (c,d)$ then $(a,b) = (a',b')$,

\item\label{itSucc} $(a,b + 1) \isuccrel_k (c,d+1)$,

\item\label{itPred} if $(c,d-1) \in U_k$ then $(a,b-1) \in U_k$ and $(a,b - 1) \isuccrel_k (c,d - 1)$, and

\item\label{itHk} $h_k(a,b) \peq h_k (c,d)$.

\end{enumerate}

\end{lemma}

\proof
We proceed by indution on $k$.
The base case, $k=0$, is proved by using the fact that ${\isuccrel_0} = \varnothing$, so the antecedent is always false. For the inductive step, let us assume that the lemma holds for all $0 \le i \le k$ and we will prove the lemma for $k+1$. To do so, let us take $(a,b),(c,d) \in \nat\times\nat$ satisfying  $(a,b) \uparrow_{k+1} (c,d)$.
If $(a,b) \isuccrel_{k} (c,d)$, the induction hypothesis immediately yields all desired properties.

Otherwise, conditions~\ref{cond:stratified:ind:1} and~\ref{cond:stratified:ind:exist} hold, so that $(x,y) \eqdef (x_k,y_k) \in U_k$ satisfies $a = x$, $c=k+1$, $b \ge y$ and $b=d$. Since $y\leq b$ we see using Lemma \ref{lemUStratifiedA}.\ref{itSClosed} that $(a,b) \in U_k \subseteq U_{k+1}$ and since also $d \geq y$ we have that $(c,d) \in U_{k+1}$ by the definition of $U_{k+1}$, establishing \eqref{itInUk}.
Moreover $a\in U_k$ so that $a\leq k$, hence $a \le k \leq k+1 = c$, so $a < c$, and by definition of $\isuccrel_{k+1}$ we must have $b=d$, establishing \eqref{itLessEq}. Since $ b < b+1$ we have that $(a,b+1),(c,d+1) \in U_{k+1}$ and $(a,b+1) \isuccrel_{k+1} (k+1,b+1) = (c,d+1)$ also by definition of $\isuccrel_{k+1}$, thus \eqref{itSucc} holds.
If $(c,d-1) \in U_{k+1}$ then $y < d =b $ so that $(a,b-1) \in U_k$, and moreover $(a,b-1) \isuccrel_{k+1} (c,d-1)$ by definition, hence \eqref{itPred}.

Finally, recall that $h_{k}(x, y) \irel v \eqdef v_k$.
Since $h_{k+1}(a, b) = h_{k+1}(x,d) = \msuccfct\Qbase^{d - y}\left(h_{k}\left(x, y\right)\right)$ and $h_{k+1}(c, d) = h_{k+1}(k+1, d) = \msuccfct^{d - y}\left(v\right)$, by the confluence condition for $\model$ and a straightforward secondary induction on $d$, $h_{k+1}(x, d) \irel h_{k+1}(c,d)$, establishing \eqref{itHk}.
\endproof

With this we may begin proving some properties of the model $\modelbase[\Qstrat]$.
We start by considering the function $h$.

\begin{lemma}\label{lemHHomoFull}
Let $\Sigma \Subset \langred$ and $\modelbase$ be any model. Then $h\colon W\Qstrat \to W$ is a function and $S\circ h = h \circ S\Qstrat$.
\end{lemma}

\proof
By Lemma \ref{lemUStratifiedA}.\ref{itHHomo}, $h_k\colon U_k \to W$ is a function for all $k$, and since $W\Qstrat = \bigcup_{k\in \mathbb N} U_k$ and $h = \bigcup_{k\in \mathbb N} h_k$ with the union being increasing, we have that $h \colon W\Qstrat \to W$.
Then we have that $ S\circ h = S\circ \bigcup_{k\in\mathbb N} h_k = \bigcup_{k\in\mathbb N} ( S\circ  h_k) =  \bigcup_{k\in\mathbb N} ( h_k \circ S\Qstrat ) =  (\bigcup_{k\in\mathbb N}  h_k ) \circ S\Qstrat = h \circ S\Qstrat $.
\endproof

\begin{lemma}\label{lem:stratified:aux}
Let $\Sigma \Subset \langred$ and $\modelbase$ be any model. Then whenever $(x,y)\peq \Qstrat (x',y') $,
\begin{enumerate}

\item\label{itStratExp} $x \leq x'$ and $y = y'$,

\item\label{itConfl} $S\Qstrat(x,y) \peq \Qstrat S\Qstrat (x',y')$,

\item\label{itConflBack} if $(x,y) = S\Qstrat (w,v)$ and $(x',y') = S\Qstrat (w',v')$ then $(w,v) \peq\Qstrat (w',v
')$, and

\item\label{itHMon} $h(x,y) \peq h (x',y')$.

\end{enumerate}
\end{lemma}

\proof
If $(x,y) \peq \Qstrat (x',y')$, then $(x,y) (\uparrow^{e})^{\star} (x',y')$.
Let $n$ in $\Nat$ and $(x_{0},y_{0}),\ldots ,(x_{n},y_{n})$ in $W\Qstrat$ be such that $(x_{0},y_{0}) = (x,y)$, $(x_{n},y_{n}) = (x',y')$ and for all nonnegative integers $i<n$, $(x_{i},y_{i}) \uparrow\Qstrat (x_{i+1},y_{i+1})$.
Thus, for all nonnegative integers $i<n$, let $k_{i}$ in $\Nat$ be such that $(x_{i},y_{i}) \uparrow_{k_{i}} (x_{i+1},y_{i+1})$.

To see that \eqref{itStratExp} holds, note that by Lemma \ref{lemUStratifiedB}.\ref{itLessEq}, for all $i<n$, $x_{i} < x_{i+1}$ and $y_{i} = y_{i+1}$.
Since $(x_{0},y_{0}) = (x,y)$ and $(x_{n},y_{n}) = (x',y')$, therefore $x\leq x'$ and $y=y'$.
For \eqref{itConfl}, by Lemma \ref{lemUStratifiedB}.\ref{itSucc} we have that for all nonnegative integers $i<n$, $(x_{i},y_{i} + 1) \uparrow_{k_{i}} (x_{i+1},y_{i+1} + 1)$, so that the sequence $((x_i,y_i+1))_{i<n}$ witnesses that $S\Qstrat (x,y) = (x,y+1) \peq\Qstrat (x',y'+1) = S\Qstrat (x',y' )$.
That \eqref{itConflBack} holds follows from similar considerations using Lemma \ref{lemUStratifiedB}.\ref{itPred}.

To establish \eqref{itHMon}, we consider the following two cases. If $n=0$, then $(x,y) = (x',y')$.
Thus $h (x,y) \peq h (x',y')$ since $\peq$ is reflexive.
Otherwise, $n\geq 1$.
Hence, by Lemma \ref{lemUStratifiedB}.\ref{itHk}, for all nonnegative integers $i<n$, $ (x_{i},y_{i} + 1) \isuccrel_{k_i}  (x_{i+1},y_{i+1} + 1)$ for all $i<n$, so that also $h (x_{i},y_{i}) \peq h  (x_{i+1},y_{i+1})$, hence by transitivity $h (x,y) \peq h (x',y')$.
\endproof

Finally, we show that $\isuccrel\Qstrat$ is suitable for producing a stratified model.

\begin{lemma}\label{lemAcyclic}
Let $\Sigma \Subset \langred$, $\modelbase$ be any model, $k\in \nat$ and $U_k,\isuccrel_k$ be as in Definition \ref{defQstrat}. Then, the graph $(W\Qstrat, \isuccrel \Qstrat)$ is acyclic and if $(0,b),(a,b) \in W\Qstrat$ there exists a unique path from $(0,b)$ to $(a,b)$.
\end{lemma}

\proof
That $(W\Qstrat, \isuccrel \Qstrat)$ is acyclic is an immediate consequence of Lemma \ref{lemUStratifiedB}.\ref{itLessEq}.
The second claim follows by induction on $a$.
Suppose that $(a,b) \in W\Qstrat$. If $a=0$ then once again by Lemma \ref{lemUStratifiedB}.\ref{itLessEq} $(0,b)$ has no predecessors and hence the singleton $  ( ( 0,b )  )$ is the unique path leading from $(0,b)$ to $(a,b)$.
Otherwise observe that if $(c,d) \isuccrel\Qstrat (a,b)$ and $(c',d') \isuccrel\Qstrat (a,b)$ then $(c,d) ,(c',d') \isuccrel_k(a,b)$ for some $k$, hence by Lemma \ref{lemUStratifiedB}.\ref{itInj} $(c,d) = (c',d') $ and by Lemma \ref{lemUStratifiedB}.\ref{itLessEq}, $d=b$. Thus by induction hypothesis there is a unique path $  ( (a_i,b_i)  )_{i<n}$ from $(0,b)$ to $(c,d)$, which means that the only path from $(0,b)$ to $(a,b)$ is $  ( (a_i,b_i) )_{i\leq n}$ with $(a_n,b_n) = (a,b)$.
\endproof

With this we are ready to show that $\model\Qstrat$ is expanding and satisfies (falsifies) the same formulae as $(\model,w)$.

\begin{lemma}
 Given $\Sigma \Subset \langred$ and a model $\M$, $\model\Qstrat$ is an expanding model.
\end{lemma}
\begin{proof}
  First we check that $\model\Qstrat$ is a model.
It is easy to see using Lemma~\ref{lem:stratified:aux}.\ref{itStratExp} that $\irel\Qstrat$ is antisymmetric, hence a partial order since it is already a transitive, reflexive closure.
  For the monotonicity condition, suppose that $(x,y) \irel\Qstrat (x',y')$.
  By Lemma~\ref{lem:stratified:aux}.\ref{itHMon}, $h(x,y) \irel\Qbase h (x',y')$ and
  by the monotonicity condition for $\model\Qbase$, 
  $V\Qstrat (x,y ) =V \big (h(x,y) \big ) \subseteq V \big (h(x',y') \big) = V\Qstrat (x',y')$.
Confluence of $S\Qstrat$ follows from Lemma \ref{lem:stratified:aux}.\ref{itConfl}.
  Therefore, $\model\Qstrat$ is a model.

  To prove that $\model\Qstrat$ is stratified,
  define $W\Qstrat_n = \ens{(x,y) \in W\Qstrat}{y = n}$
  for all $n \in \nat$.
  Condition~\ref{cond:stratified:succ} of Def.~\ref{def:stratified} trivially holds,
  condition~\ref{cond:stratified:closed} comes directly from Lemma~\ref{lem:stratified:aux}.\ref{itStratExp}, and condition~\ref{cond:stratified:tree} from Lemma \ref{lemAcyclic}.
Moreover, $\model\Qstrat$ is expanding by Lemma \ref{lem:stratified:aux}.\ref{itConflBack}.
\end{proof}

\begin{lemma}
Let $\Sigma \Subset \langred$ and $\modelbase$ be any model.
  For any state $(x,y) \in W\Qstrat$ and any $\psi \in \Sigma$, 
  $\model\Qstrat, (x,y) \sat \psi$ if and only if $\model\Qbase, h (x,y) \sat \psi$.
\end{lemma}
\begin{proof}
  The proof is by induction on the size $\nos{\psi}$ of the formula.
  The cases for propositional variables, falsum, conjunctions and disjunctions are straightforward.
  For the temporal modalities, recall that
  for all $(x, y) \in W\Qstrat$ and all $n \in \nat$,
  $(S\Qstrat )^ n(x,y) = (x, y + n) \in W\Qstrat$, so that by Lemma \ref{lemHHomoFull}, $h (x, y + n) = \msuccfct\Qbase^n\left(h (x,y)\right)$, which allows us to easily apply the induction hypothesis.
  
  Finally, for implication, suppose first that $\model\Qstrat, (x,y) \nsat \psi_1 \imp \psi_2$.
  Then there is $(x',y')$ such that $(x, y) \irel\Qstrat (x', y')$,
  $\model\Qstrat, (x', y') \sat \psi_1$ and $\model\Qstrat, (x', y') \nsat \psi_2$.
  By Lemma~\ref{lem:stratified:aux}.\ref{itHMon}, $h (x, y) \irel\Qbase h (x', y')$
  and by induction hypothesis, $\model\Qbase, h (x', y') \sat \psi_1$ and
  $\model\Qbase, h (x', y') \nsat \psi_2$.
  Therefore, $\model\Qbase, h (x, y) \nsat \psi_1 \imp \psi_2$.
  For the other direction suppose that $\model\Qbase, h (x,y) \nsat \psi_1 \imp \psi_2$.
  Hence, There is $v' \in W\Qbase$ such that $h (x, y) \irel\Qbase v'$,
  $\model\Qbase, v' \sat \psi_1$ and $\model\Qbase, v' \nsat \psi_2$.
  Let $k$ be such that $(x_k,y_k,H_k) = (x, y, \Sigma(v'))$; then, $v'$ witnesses that~\ref{cond:stratified:ind:exist} holds, and since $x \le k$,
  condition~\ref{cond:stratified:ind:belong} holds too.
  Hence, there is $(x', y') \in W\Qstrat$ such that $\Sigma(h (x', y')) = \wsf{v'}$
  and $(x, y) \isuccrel\Qstrat (x', y')$, which implies that $(x, y) \irel \Qstrat (x', y')$.
  By induction hypothesis, $\model\Qstrat, (x',y') \sat \psi_1$ and
  $\model\Qstrat, (x', y') \nsat \psi_2$, hence $\model\Qstrat, (x, y) \nsat \psi_1 \imp \psi_2$.
\end{proof}

In conclusion, we obtain the following:

\begin{theorem}\label{TheoStrat}
  A formula $\varphi$ is satisfiable (resp. falsifiable) on an intuitionistic dynamic model
  if and only if   it is satisfiable (resp. falsifiable) on an expanding model.
\end{theorem}

\section{Special classes of frames}\label{secSpecial}

As we have seen in Propositon \ref{PropIntCond}, the class of dynamic posets is the widest class of posets equipped with a function that satisfy truth monotonicity under the classical interpretation of the temporal modalities.
However, in the literature one often considers smaller classes of frames.
In this section we will discuss persistent and here-and-there models, and compare their logics to $\iltl$.

\subsection{Persistent frames}

Expanding models were introduced as a weakening of product models, and thus it is natural to also consider a variant of $\iltl$
interpreted over `standard' product models,
or over the somewhat wider class of persistent models.

\begin{definition}
Let $(W,{\preccurlyeq})$ be a poset.
If $\msuccfct\colon W\to W$ is such that,
whenever $v \succcurlyeq \sfun(w)$, there is $u\succcurlyeq w$ such that $v=\sfun (u)$,
we say that $\sfun$ is {\em backward confluent}.
If $\sfun$ is both forward and backward confluent, we say that it is {\em persistent}.
A tuple $(W,{\preccurlyeq},\msuccfct)$ where $\msuccfct$ is persistent is a {\em persistent
intuitionistic temporal frame}, and the set of valid formulas over the class of persistent
intuitionistic temporal frames is denoted $\itlb$, or {\em persistent domain $\sf ITL$}.
\end{definition}

See Figure \ref{FigCO} for an illustration of backwards confluence.
The name `persistent' comes from the fact that Theorem \ref{TheoStrat} can be modified to obtain a stratified model $\model'$ where $\msuccfct'\colon W'_k\to W'_{k+1}$ is an isomorphism, i.e.~whose domains are persistent with respect to $\msuccfct'$, although we will not elaborate on this issue here.
Next we remark that $\iltl\subsetneq \itlb$, given the following claim proven in \cite{BoudouJELIA}.
\begin{proposition}\label{prop:nvalid:iltl}
The formula $\left(\tnext \varphi \rightarrow \tnext \psi \right) \imp \tnext \left(\varphi \rightarrow \psi\right)$ is not $\iltl$-valid. However it is $\itlb$-valid.
\end{proposition}

Over the class of persistent models this property will allow us to `push down' all occurrences of $\tnext$ to the propositional level.
Say that a formula $\varphi$ is in {\em $\tnext$-normal form} if all occurrences of $\tnext$ are of the form $\tnext ^i p$, with $p$ a propositional variable.

\begin{theorem}\label{TheoNextInside}
Given $\varphi\in \lang$, there exists $\widetilde\varphi$ in $\tnext$-normal form such that $\varphi\leftrightarrow\widetilde\varphi$ is valid over the class of persistent models.
\end{theorem}

\proof
The claim can be proven by structural induction using the validities in Propositions \ref{prop:valid:iltl}, \ref{prop:nvalid:iltl} and \ref{PropUValid}.
\endproof

We remark that the only reason that this argument does not apply to arbitrary $\iltl$ models is the fact that $(\tnext\varphi\to\tnext\psi)\to\tnext(\varphi\to\psi)$ is not valid in general (Proposition \ref{prop:nvalid:iltl}).
Next we show that the finite model property fails over the class of persistent models, using the following formula.

\begin{lemma}\label{LemmInfExam}
The formula $\varphi=\neg\neg\diam\ubox p\to \diam\neg\neg\ubox p$ is not valid over the class of persistent models.
\end{lemma}

\proof
Consider the model $ M=(W,{\preccurlyeq},\msuccfct,V)$, where $W=\mathbb Z \cup \{r\}$ with $r$ a fresh world not in $\mathbb Z$, $w\preccurlyeq v$ if and only if $w=r$ or $w=v$, $\msuccfct(r)=r$ and $\msuccfct(n)=n+1$ for $n\in\mathbb Z$, and $\llbracket p\rrbracket =[0,\infty)$. It is readily seen that $ M$ is a persistent model, that $\model,r\sat \neg\neg\diam\ubox p$ (since every world above $r$ satisfies $\diam\ubox p$), yet $\model,r\nsat \diam\neg\neg\ubox p$, since there is no $n$ such that $\model, \msuccfct^n(r) \sat \neg\neg\ubox p$. It follows that $\model,r\nsat  \varphi $, and hence $\varphi$ is not valid, as claimed.
\endproof

\begin{lemma}\label{LemmFinExam}
The formula $\varphi$ (from Lemma \ref{LemmInfExam}) is valid over the class of finite, persistent models.
\end{lemma}

\proof
Let $\model=(W,{\preccurlyeq},\msuccfct,V)$ be a finite, persistent model, and assume that $\model,w\sat \neg\neg\diam\ubox p $.
Let $v_1,\hdots,v_n$ enumerate the maximal elements of $\{v\in W \mid w\peq v\}$. 
For each $i\leq n$, let $k_i$ be large enough so that $\model,\msuccfct^{k_i}(v_i)\sat \ubox p $, and let $k=\max k_i$. We claim that $\model, \msuccfct^k(w)\sat \neg\neg\ubox p $, which concludes the proof. Let $u\succcurlyeq \msuccfct^k(w)$ be any leaf. Then, there is $v' \succcurlyeq w$ such that $u=\msuccfct^{k}(v')$ (since compositions of persistent functions are persistent). Choosing a leaf $v \seq v'$, we obtain by forward confluence of $S^k$ that $S^k(v) =  u$ (as $u$ is already a leaf). But, since $k\geq k_i$, we obtain $\model,u\sat \ubox p$. Since $u$ was arbitrary we easily obtain $\model,w\sat \diam \neg \neg \ubox p$, as desired.
\endproof

The following is then immediate from Lemmas \ref{LemmInfExam} and \ref{LemmFinExam}:

\begin{theorem}\label{ThmNonFin}
${\itlb}$ does not have the finite model property.
\end{theorem}

Thus our decidability proof for $\iltl$, which proceeds by first establishing an effective finite model property, will not carry over to $\itlb$. Whether $\itlb$ is decidable remains open.

\subsection{Temporal here-and-there models}\label{SecHT}

An even smaller class of models which, nevertheless, has many applications is that of temporal here-and-there models \cite{CP07,BalbianiDieguezJelia}. Some of the results we will present here apply to this class, so it will be instructive to review it. The logic of here-and-there is the maximal logic strictly between classical and intuitionistic propositional logic, given by a frame $\{0,1\}$ with $0 \peq 1$. This logic is axiomatized by adding to intuitionistic propositional logic the axiom $p \vee ( p \imp q )\vee \neg q.$

A temporal here-and-there frame is a persistent frame that is `locally' based on this frame. To be precise:

\begin{definition}
A {\em temporal here-and-there frame} is a persistent frame $(W,{\peq},S)$ such that $W = T \times \{0,1\}$ for some set $T$, and there is a function $f\colon T\to T$ such that for all $t,s \in T$ and $i,j\in \{ 0,1\}$, $(t,i) \peq (s,j)$ if and only if $t=s$ and $i\leq j$ and $S(t,i) = (f(t),i)$.
\end{definition}

The prototypical example is the frame $(W,{\peq},S)$, where $W = \mathbb N\times \{0,1\}$, $(i,j) \peq (i',j')$ if $i = i'$ and $j \leq j'$, and $S(i,j) = (i+1,j)$. Note, however, that our definition allows for other examples (see Figure \ref{FigBoxU}). We will denote the resulting logic by $\itlht$. In its propositional flavour, here-and-there logic plays a crucial role in the definition of Equilibrium Logic~\cite{Pea96,Pea06}, a well-known characterisation of Stable Model~\cite{Gelfond88} and Answer Set~\cite{Niemela99,MT99} semantics for logic programs. Modal extensions of this aforementioned superintuitionistic logic made it possible to extend those existent logic programming paradigms with new constructs, allowing their use in different scenarios where describing and reasoning with temporal~\cite{CP07} or epistemic~\cite{CerroHS15} data is necessary. A combination of propositional here-and-there with {$\sf LTL$} was axiomatized by Balbiani and Di\'eguez \cite{BalbianiDieguezJelia}, who also show that $\ubox$ cannot be defined in terms of $\diam$, a result we will strengthen here to show that $\ubox$ cannot be defined even in terms of $\Until$. It is also claimed in \cite{BalbianiDieguezJelia} that $\diam$ is not definable in terms of $\ubox$ over the class of here-and-there models, but as we will see in Proposition \ref{propDiamDefinHT}, this claim is incorrect.

\section{Combinatorics of intuitionistic models}\label{sec:combinatorics}

In this section we introduce some combinatorial tools we will need in order to prove that $\iltl$ has the effective finite model property, and hence is decidable. We begin by discussing labelled structures, which allow for a graph-theoretic approach to intuitionistic models.

\subsection{Labelled structures and quasimodels}

\begin{definition}
Given a set $\Labels$ whose elements we call `labels' and a set $W$, a {\em $\Lambda$-labelling
function} on $W$ is any function $\labfun\colon W\to \Labels$. A structure ${
S}=(W,R,\labfun)$ where $W$ is a set, $R\subseteq W\times W$ and $\labfun$ is a labelling function on
$W$ is a {\em $\Labels$-labelled structure,} where `structure' may be replaced with `poset', `directed graph', etc.
\end{definition}

A useful measure of the complexity of a labelled poset or graph is given by its level:

\begin{definition}
Given a labelled poset $\A=(W, {\peq},\labfun)$ and an element $w \in W$,
an \emph{increasing chain} from $w$ of length $n$ is
a sequence $v_1 \ldots v_n$ of elements of $W$ such that $v_1 = w$ and $\forall i < n ,~ v_i \prec v_{i+1},$
where $u\prec u'$ is shorthand for $u\peq u'$ and $u'\not\peq u$. The chain $v_1 \ldots v_n$ is {\em proper} if it moreover satisfies $\forall i < n ,~ \labfun\left(v_i\right) \ne \labfun\left(v_{i+1}\right).$
The {\em depth} $\dpt{w} \in \nat \cup \{\omega\}$ of $w$ is defined such that
$\dpt{w} = m$ if $m$ is the maximal length of all the increasing chains from $w$
and $\dpt{w} = \omega$ if there is no such maximum.
%
Similarly,
the \emph{level} $\degp{w} \in \nat \cup \{\omega\}$ of $w$ is defined such that
$\degp{w} = m$ if $m$ is the maximal length of all the proper increasing chains from $w$
and $\degp{w} = \omega$ if there is no such maximum.
The level $\degp{\A}$ of $\A$ is the maximal level of all of its elements.

The notions of {\em depth} and {\em level} are extended to any acyclic directed graph $(W, {\uparrow},\labfun)$ by taking the respective values on $(W, {\uparrow}^\ast,\labfun)$.
\end{definition}

An important class of labelled posets comes from intuitionistic models. 
Below, recall that $\Sigma_\model(w) = \ens{\psi \in \Sigma}{\model, w \sat \psi}$, and we may omit the subindex `$\model$'.

\begin{definition}
Given an intuitionistic Kripke model $\model = (W, \irel, V)$, we denote the labelled poset $(W, \irel, \Sigma_\model)$ by $\model^\Sigma$.
Conversely,
given a labelled poset $\A=(W,{\peq},\labfun)$ over $\parts\Sigma$ such that if $w \irel v$ then $\labfun(w) \subseteq \labfun(v)$,
the valuation $V_{\labfun}$ is defined such that
$V_{\labfun}(w) = \{p \in \Var \mid p \in \labfun(w) \}$ for all $w \in W$, and denote the resulting model by $\newmodel\A$.
\end{definition}

If $\model = (W, \irel, V)$ is a model, it can easily be checked that for all $w, v \in W$,
if $w \irel v$ then $\Sigma(w) \subseteq \Sigma(v)$.
Note that not every $\parts\Sigma$-labelled poset is of the form $\model^\Sigma$, as it has to satisfy additional conditions according to the semantics.
In particular, we are interested in labelled posets that respect the intuitionistic implication:

\begin{definition}
Let $\Sigma\Subset \lang_{\ubox\Until}$ and $\A=(W,{\peq},\labfun)$ be a $\parts\Sigma$-labelled poset. We say that $\A$ is a {\em $\Sigma$-quasimodel}\david{Note the clash of terminology with \cite{DFD2016}. We might consider using different terminology.} if $\labfun $ is monotone in the sense that $w\peq  v$ implies that $\labfun(w)\subseteq \labfun(v)$, and whenever $\varphi\to\psi\in \Sigma$ and $w\in W$, we have that $\varphi\to\psi \in \labfun (w)$ if and only if, for all $v$ such that $w\peq v$, if $\varphi\in\labfun(v)$ then $\psi\in\labfun(v)$.

If further $(W,{\peq})$ is a tree, we say that $\A$ is {\em tree-like.}
\end{definition}

\subsection{Simulations, immersions and condensations}

As is well-known, truth in intuitionistic models is preserved by bisimulation, and thus this is usually the appropriate notion of equivalence between different models.
However, it will also be convenient to consider a weaker notion, which we call {\em bimersion}.

\begin{definition}\label{DefSim}
Given two labelled posets $\A=(W_\iA,\peq_\iA,\labfun_\iA)$ and $\B =(W_\iB ,\peq_\iB ,\labfun_\iB )$
and a relation $R\subseteq  W_\iA\times W_\iB $,
we write
\begin{align*}
\dom(R) & = \ens{w\in W_\iA}{\exists v \in W_\iB  ~ (w,v) \in R}\\
\rng(R) & =  \ens{v \in W_\iB }{\exists w \in W_\iA ~ (w,v) \in R}.
\end{align*}
A relation ${\simvar} \subseteq W_\iA \times W_\iB $ is
a \emph{simulation} from $\A$ to $\B $ if
$\dom({\simvar})=W_\iA$ and
whenever $w\mathrel\simvar v$,
it follows that $\labfun_\iA(w)=\labfun_\iB  (v)$,
and if $w\peq_\iA w'$ then there is $v'$ so that $v\peq_\iB  v'$ and $w'\mathrel \simvar v'$.

A simulation is called a {\em (partial) immersion} if it is a (partial) function.
If an immersion $\simvar\colon W_\iA\to W_\iB $ exists, we write $\A\imm \B $.
If, moreover, there is an immersion $\tau\colon W_\iB \to W_\iA$,
we say that they are {\em bimersive},
write $\A\simm \B $, and call the pair $(\simvar,\tau)$ a {\em bimersion}.
A {\em condensation from $\A$ to $\B$} is a bimersion $({\srone},{\srtwo})$
so that $\srone\colon W_\iA\to W_\iB $, $\srtwo\colon W_\iB \to W_\iA$,
$\srone$ is surjective, and $\srone\srtwo$ is the identity on $W_\iB$.
If such a condensation exists we write $\B \cond \A$.
Observe that $\B \cond \A$ implies that $\B \simm \A$. 

If $\M ,\N $ are models and $\Sigma\Subset \langred$, we write $\M \imm_\Sigma\N $ if $\M ^\Sigma\imm\N ^\Sigma$, and define $\simm_\Sigma,\cond_\Sigma$ similarly. We may also write e.g.~$\A\cond \M $ if $\A$ is $\parts\Sigma$-labelled and $\A\cond\M ^\Sigma$. 
\end{definition}

\begin{figure}\centering
\begin{tikzpicture}[node distance=2cm] 

  \node (a10) {$\varnothing$};
  \node (c11)[above left of=a10] {$\lbrace \tnext p\rbrace$};
  \node (a12)[above right of=a10] {$\varnothing$};
  \node (b13)[above left of=a12] {$\lbrace \tnext p\rbrace$};
  \node (c14)[above right of=a12] {$\lbrace p \rbrace$};

  \node (a20)[right=6cm of a10] {$\varnothing$};
  \node (b21)[above left of=a20] {$\lbrace \tnext p\rbrace$};
  \node (c22)[above right of=a20] {$\lbrace p \rbrace$};

  \path[intuitionistic relation]
  (a10) edge[] node {} (c11)
  (a10) edge[] node {} (a12)
  (a12) edge[] node {} (b13)
  (a12) edge[] node {} (c14)
  (a20) edge[] node {} (b21)
  (a20) edge[] node {} (c22)
  ;

  \path[dotted, thick, ->, >=stealth']
  (a10) edge[bend left=10] node {} (a20)		
  (a20) edge[bend left=10] node {} (a10)	
  (c11) edge[bend left=10] node {} (b21)		
  (b13) edge[bend left=5] node {} (b21)
  (b21) edge[bend left=5] node {} (b13)	
  (c14) edge[bend left=10] node {} (c22)		
  (c22) edge[bend left=10] node {} (c14)
  (a12) edge[bend left=10] node {} (a20)
  ;
\end{tikzpicture}				
\caption{A condensation from the labelled frame on the left to the labelled frame on the right.
  Dotted arrows indicate the condensation: $\rho$ for arrows from left to right and $\iota$ for arrows from right to left.}
\label{fig:condensate-frame}
\end{figure}
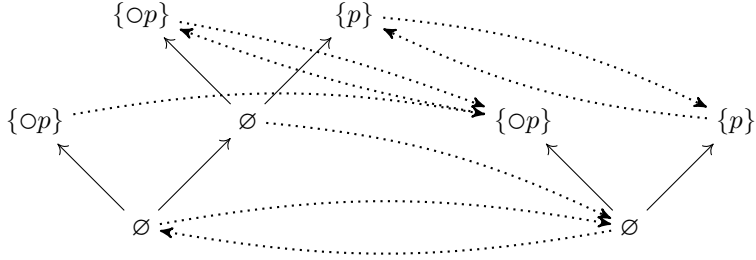

See Figure \ref{fig:condensate-frame} for an example of a condensation.
Note that the relation $\simm$ is an equivalence relation.
In this text, simulations will {\em always} be between posets. In the case that $\A$ or $\B$ is an acyclic directed graph, a simulation between $\A$ and $\B$ will be one between their respective transitive, reflexive closures. It will typically be convenient to work with immersions rather than simulations: however, as the next lemma shows, not much generality is lost by this restriction.

\begin{lemma}\label{LemmThereIsFun}
  Let $\fulllab\A\iA$ and $\fulllab\B\iB$ be labelled posets.
  If a simulation ${\simvar}\subseteq W_\iA\times W_\iB$ exists,
  $W_\iA$ is a finite tree, and $w\mathrel\simvar w'$,
  then there is a partial immersion ${\simvar}'\colon W_\iA \to W_\iB$ such that $w\in \dom({\simvar'})$ and $w'=\simvar'(w)$.
\end{lemma}

\proof
By a straightforward induction on the depth of $w \in W_ \iA$ we show that if $w\mathrel\simvar w'$ then there is a partial immersion $\simvar_w$ with $w\in \dom ({\simvar_w})$, whose domain is the subtree generated by $w$, and such that $\simvar_w(w) = w'$.
Let $D$ be set of daughters of $w$,
and for each $v\in D$, choose $v'$ so that $v\mathrel \simvar v'$ and $w'\peq_\iB v'$.
By the induction hypothesis, there is a partial immersion $\simvar_{v}'$ with $v\in \dom(\simvar_{v}')$.
Then, one readily checks that $\{(w,w')\}\cup\bigcup_{v\in D}\simvar_{v}'$
is also an immersion, as needed. 
\endproof

Condensations are useful for producing (small) quasimodels out of models.

\begin{proposition} \label{prop:IL}
  Given
  an intuitionistic model $\fullmod\M \iM $,
  a set $\Sigma \Subset \langred$,
  and a $\parts\Sigma$-labelled poset $\fulllab\A\iA$ over $\Sigma$,
  if $\A \cond \M$,
  then $\A$ is a quasimodel.
\end{proposition}

\begin{proof}
Let $(\srone,\srtwo)$ be a condensation from $\model^\Sigma$ to $\A$. If $w\peq_\iA v $, then $\srtwo(w)\peq_\iM \srtwo(v)$, so that $\labfun_\iA(w)=\Sigma(\srtwo(w))\subseteq \Sigma( \srtwo(v))=\labfun_\iA(v)$. Next, suppose that $\varphi\to\psi\in \labfun_\iA(w)$, and consider $v$ such that $w\peq_\iA v$. Then, $\M , \srtwo(w) \sat \varphi \imp \psi$. Since $\srtwo$ is an immersion, $\srtwo(w)\peq_\iM \srtwo(v)$, hence if $\M , \srtwo(v) \sat \varphi  $, then also $\M , \srtwo(v) \sat   \psi$. Thus if $\varphi\in\labfun_\iA(v)$, it follows that $\psi\in\labfun_\iA(v)$.
Finally, suppose that $\varphi\to\psi \in \Sigma\setminus \labfun_\iA(w)$. Then, $\M,\srtwo(w)\nsat \varphi\to\psi$, so that there is $v \in W_\iM$ such that $\srtwo(w) \irel_\iM v$,
  $\M, v \sat \varphi$ and $\M, v \nsat \psi$.
  It follows that $\varphi \in \labfun_\iA(\srone (v))$ and $\psi \not\in \labfun_\iA(\srone (v))$, and since $\srone$ is an immersion we also have that $w=\srone\srtwo(w)\peq_\iA \srone(v)$, as needed.
\end{proof}

\subsection{Normalized labelled trees}

In order to
count the number of different labelled trees up to bimersion,
we construct,
for any set $\Labels$ of labels and any $k \ge 1$,
the labelled directed acyclic graph $\bgraphbase{k}$ by induction on $k$ as follows.

\paragraph*{Base case.}
For $k = 1$, let $\bgraphbase 1$ with $W^\Labels_{1} = \Labels$, $\isuccrel^\Labels_{1} = \varnothing$, and $\labfun^\Labels_{1}(w) = w $ for all $ w \in W^\Labels_{1}.$

\paragraph*{Inductive case.}
Suppose that $\bgraphbase{k}$ has already been defined. Let us write $X \amalg Y$ for the disjoint union of $X$ and $Y$. The graph $\bgraphbase{k+1}$ is constructed such that:
\begin{align*}
  W^\Labels_{k+1} &= W^\Labels_{k} \amalg \tilde W^\Labels_{k+1} \text{, where } \tilde W^\Labels_{k + 1}= 
   \Labels \times \parts{W^\Labels_{k}}  \\
  \isuccrel^\Labels_{k+1} &= \isuccrel^\Labels_{k} \cup
    \ens{((\ell,C),y) \in \tilde W^\Labels_{k + 1} \times W^\Labels_{k}}{y \in C} \\
  \labfun^\Labels_{k+1}(w) &= \begin{cases}
      \labfun^\Labels_{k}(w) &\text{if } w \in W^\Labels_{k} \\
      \ell       &\text{if } w = (\ell, C) \in \tilde W^\Labels_{k+1}
    \end{cases}
    \end{align*}\david{Before we had
    \[W^\Labels_{k+1} = W^\Labels_k \cup \ens{(\ell, C) \in \Labels \times \parts{W^\Labels_{k}}}{
    \forall y \in C,~ \labfun^\Labels_{k}(y) \ne \ell}.\]
    As far as I can tell this simpler definition works just as well.}

Note that $\bgraphbase k$ is typically not a tree, but we may unravel it to obtain one.

\begin{definition}
Given a  labelled directed graph $\Grph  = (W, \isuccrel, \labfun)$ and $w\in W$,
the \emph{unravelling} of $\Grph $ from $w$ is
the labelled tree $\ur w \Grph  =(\ur w W, \ur w\isuccrel, \ur w\labfun)$ such that
$\ur w W$ is the set of all the paths in $\Grph$ starting on $w$,
$\xi \mathrel{\ur w \isuccrel}\zeta $ if and only if there is $v \in W $
such that $ \zeta = \xi v$, and $\ur w \labfun(v_0 \ldots v_n) = \labfun(v_n)$.
\end{definition}

\begin{proposition} \label{PropIsBisim}
  For any rooted labelled tree
  $\T $ over a set $\Labels$ of labels,
  if the level of $\T$ is finite then
  there is a condensation from $\T$ to $\ur y { \Grph^\Lambda_{{\rm lev}(\T)} }$ for some $y \in W^\Lambda_{{\rm lev}(\T)}$.
\end{proposition}

\begin{proof}
  Let $\T = (W_\T, \isuccrel_\T, \labfun_\T)$ be a labelled tree with root $r$.
  We write $\irelneq_\T$
  for the transitive closure of $\isuccrel_\T$
  and $\irel_\T$
  for the reflexive closure of $\irelneq_\T$.
  The proof is by induction on the level $n = \degp \T$ of $\T$.
  For $n = 1$, observe that this means that $\labfun_\T(w)=\labfun(r)$ for all $w\in W_\T$.
  Let $\srone = W_\T\times \{\labfun_\T(r)\}$ and $\srtwo=\{(\labfun_\T(r),r)\}$.
  It can easily be checked that
  $(\srone,\srtwo)$ is a condensation.\footnote{Recall that as per our convention, this means that $(\srone,\srtwo)$ is a condensation between the respective transitive closures.}
  For $n > 1$, suppose the property holds for all rooted labelled trees $\T'$ such that
  $\degp {\T'} < n$.
  Define the following sets:
  \begin{align*}
    N &= \ens{w \in W_\T}{\labfun_\T(w) \neq \labfun_\T(r)
            \text{ and for all } v \irelneq w ,~ \labfun_\T(v) = \labfun_\T(r)} \\
    M &= \ens{w \in W}{\text{for all } v \irel w ,~ \labfun_\T(v) = \labfun_\T(r)}
  \end{align*}
  Note that if $w \in N$ then $\degp w < \degp r$, and therefore $\degp w < n$; hence by induction, there is a condensation $({\srone}'_w,{\srtwo}'_w)$
  from the subgraph of $\T$ generated by $w$
  to $\ur {y_w} { \Grph^\Labels_{n-1} }$ for some $y_w \in W ^\Labels_{n-1}$.

 Define $s = (\labfun(r), \ens{y_w}{w \in N})\in W^\Lambda_{n}$ and consider the unravelling $\mathfrak U = (W_\mathfrak U,\isuccrel_\mathfrak U,\lambda_\mathfrak U)$ of $\Grph^\Labels_n$ from $s$.
Note that $\ur{y_w}{ \Grph^\Labels_{n-1} }$ embeds into $\mathfrak U$ via the map $\xi \mapsto s \xi$, and with this we define $\rho_w \colon W_\T \to W_\mathfrak U$ by $\rho_w = s \rho'_w$, and similarly define $\iota_w \colon W_\mathfrak U \to W_\T$ by $\iota_w (s\xi) = \iota'_w (\xi)$ (i.e., $\iota_w$ first removes the first element of a string and then applies $\iota'_w$).

We then define
\begin{align*}
{\srone} & = \left(M \times \left\{s\right\} \right) \cup \bigcup_{w \in N} {\srone}_w,\\
{\tilde \srtwo} & = \{(s,r)\} \cup \bigcup_{w \in N} {\srtwo}_w.
\end{align*}
Then, it can readily be checked that $\srone$
  is an immersion from $\T$ to~$\mathfrak U$,
$\tilde \srtwo$  is a simulation from $\mathfrak U$ to~$\T$
  and $\tilde \srtwo\subseteq \srone^{-1}$.
  Using Lemma~\ref{LemmThereIsFun}, we can then choose an immersion $\srtwo\subseteq\tilde \srtwo$, so
  that $(\srone,\srtwo)$ is a condensation from $\T$ to~$\mathfrak U$.\philippe{I do not agree that this can be easily checked. We should take into account that the reader is maybe lost in the huge set of notations and definitions that she has to keep in mind. If it is explicitly explained why $\rho = (M \times {s}) \cup \bigcup_{w \in N} \rho_{w}$, then she will understand better.}
\david{That $\rho = (M \times {s}) \cup \bigcup_{w \in N} \rho_{w}$ is a definition, not a claim. Is it clearer now? If not, we can add more detail.}
\end{proof}

\newcommand{\nre}[2]{E^{#1}_{#2}}
\newcommand{\nrq}[2]{Q^{#1}_{#2}}

Finally, given $n,k\in\nat$ let us recursively define natural numbers $\nre{n}{k}$ and $\nrq{n}{k}$ by:

\begin{equation*}
  \nre{n}{k} = \begin{cases}
    0 &\text{if } k = 0 \\
    \nre{n}{k-1} + n 2^{\nre{n}{k-1}} &\text{otherwise}
  \end{cases}
  \qquad
  \nrq{n}{k} = \begin{cases}
    0 &\text{if } k = 0 \\
    1 + \nre{n}{k-1} \nrq{n}{k-1} &\text{otherwise}
  \end{cases}
\end{equation*}

The following lemma can be proven by a straightforward induction, left to the reader.\philippe{Again, I do not agree with ``straightforward'' and ``left to the reader''. As for me, I was unable to figure out the meaning of all these notations.}

\begin{lemma}\label{LemmBoundsGraph}

  For any finite set $\Labels$ with cardinality $n$ and all $k \in \nat$,
\begin{enumerate*}

\item
  the size of $\bgraph{k}$ is bounded by $\nre{n}{k}$, and
  
\item 
  the size of any unravelling of $\bgraph k$
  is bounded by $\nrq{n}{k}$.

\end{enumerate*}

\end{lemma}

From this and Proposition \ref{PropIsBisim}, we obtain the following:

\begin{theorem}\label{TheoKruskal}

\begin{enumerate}

\item 
Given a set of labels $\Labels$ and a $\Labels$-labelled tree $\T$ of level $k<\omega$, there is a $\Labels$-labelled tree $\T'$ bounded by $\nrq{|\Labels|}{k}$ such that $\T'\simm \T$. We call $\T'$ the {\em normalized $\Labels$-labelled tree for $\T$.}

\item \label{ItKruskal} Given a sequence of $\Labels$-labelled trees $\T_1,\hdots,\T_n$ of level $k<\omega$ with $n > \nre{|\Labels|}{k}$, there are indexes $i<j\leq n$ such that $\T_i\simm \T_j$.

\end{enumerate}

\end{theorem}

\begin{proof}
In view of Proposition \ref{PropIsBisim}, way may take $\T$ to be a suitable unravelling of $\Grph^\Lambda_k$, establishing the first claim. For the second, by Lemma \ref{LemmBoundsGraph}, $ \Grph^\Lambda_k$ has size at most $\nre{|\Lambda|}k$. Since the unravellings of any graph are determined by their starting point, there must be $i<j\leq n$ with $\T_i$ and $\T_j$ bimersive to the same unravelling of $\Grph^\Lambda_k $, from which it follows that $\T_i$ and $\T_j$ are bimersive.
\end{proof}

The second item may be viewed as a finitary variant of Kruskal's theorem for labelled trees \cite{Kruskal1960}. When applied to quasimodels, we obtain the following:

\begin{proposition}\label{PropBound}

Let $\Sigma \Subset \langred$ with $|\Sigma|=s <\omega$.

\begin{enumerate}

\item 
Given a tree-like $\Sigma$-quasimodel $\T$, there is a tree-like $\Sigma$-quasimodel $\T'\simm_{\Sigma}\T$ bounded by $\nrq{2^{s}}{s+1}$. We call $\T'$ the {\em normalized $\Sigma$-quasimodel for $\T$.}

\item Given a sequence of tree-like $\Sigma$-quasimodels $\T_1,\hdots,\T_n$ with $n > \nre{2^{s}}{s+1}$, there are indexes $i<j\leq n$ such that $\T_i\simm\T_j$.

\end{enumerate}
\end{proposition}

\proof
Immediate from Proposition \ref{prop:IL} and Theorem~\ref{TheoKruskal} 
using the fact that any
$\Sigma$-quasimodel has level at most $s+1$.
\endproof

Finally, we obtain an analogous result for pointed structures.

\begin{definition}
A {\em pointed labelled poset} is a structure $(W,\peq,\labfun,w)$ consisting of a labelled tree with a designated world $w\in W$. Given a labelled poset $\A=(W_\iA,\peq_\iA,\labfun_\iA)$ and $w\in W_\iA$, we denote by $\A^w$ the pointed labelled poset given by $\A^w = (W_\iA,\peq_\iA,\labfun_\iA, w)$. A {\em pointed simulation} between pointed labelled posets $\A=(W_\iA,\peq_\iA,\labfun_\iA,w_\iA)$ and $\B=(W_\iB,\peq_\iB,\labfun_\iB,w_\iB)$ is a simulation $\simvar\subseteq W_\iA\times W_\iB$ such that if $w \mathrel \simvar v$, then $w=w_\iA$ if and only if $v=w_\iB$. The notions of {\em pointed immersion,} {\em pointed condensation,} etc.~are defined analogously to Definition \ref{DefSim}.
\end{definition}

\begin{lemma}\label{LemmPointedBound}
If $\Labels$ has $n$ elements, any pointed $\Labels$-labelled poset of level at most $k$ condenses to a labelled pointed tree bounded by $\nrq{2n}{k+2}$, and there are at most $\nre{2n}{k+2}$ bimersion classes.
\end{lemma}

\proof
We may view a pointed labelled poset $\A=(W,\peq,\labfun,w)$ as a (non-pointed) labelled poset as
follows. Let $\Labels'=\Labels\times \{0,1\}$. Then, set $\labfun'(v)=(\labfun(v),0)$ if $v\ne w$,
$\labfun'(w)=(\labfun(w),1)$.
Note that if $\A$ had level $k$ according to $\lambda$ it may now have level $k+2$ according to $\lambda'$, since if $u\prec w\prec v$ we may have that $\labfun(u)=\labfun(w)=\labfun(v)$ yet $\labfun'(u)\ne\labfun'(w)$ and $\labfun'(w)\ne\labfun'(v)$. By Proposition \ref{PropIsBisim}, $\A$ condenses to a generated tree $\T$ of $\Grph^{\Labels'}_{k+2}$
 by some condensation $({\srone},{\srtwo})$. Let $w'=\srone(w)$, and consider $\T$ as a pointed structure with distinguished point $w'$. Given that $\srone$ is a surjective, label-preserving function, $w,w'$ are the only points whose label has second component $1$, and therefore $({\srone},{\srtwo})$ must be a pointed condensation, as claimed.
\endproof

With this we may give an analogue of Proposition \ref{PropBound} tailored for pointed quasimodels.
Its proof is essentially the same.

\begin{proposition} \label{PropBoundPointed}
  Let $\Sigma \Subset\langred$ with $|\Sigma|=s $.
  \begin{enumerate}
    \item Given a tree-like pointed $\Sigma$-quasimodel $\T$
          and a formula $\varphi$,
          there is a tree-like pointed $\Sigma$-quasimodel $\T'\simm \T$ bounded by $\nrq{2^{s+1}}{s+3}$.
          We call $\T'$ the {\em normalized pointed $\Sigma$-quasimodel for $\T$.}

    \item Given a sequence of tree-like pointed $\Sigma$-quasimodels $\T_1,\hdots,\T_n$ with $n > \nre{2^{s+1}}{s+3}$,
          there are indexes $i<j\leq n$ such that $\T_i\simm\T_j$.
  \end{enumerate}
\end{proposition}

With these tools at hand, we are ready to prove that $\iltl$ has the effective finite model property, and hence is decidable.


\section{The Finite Model Property} \label{SecDec}

In view of Proposition \ref{propExpEquiv}, in order to show that validity over $\lang$ is decidable, it suffices to prove that validity is decidable over $\lang_{\mathord{\ubox}\mathord{\Until}}$.
Thus in this section we will restrict our attention to this sub-language.
We will use the notions of {\em eventuality} and {\em fulfilment,} defined below (see also Figure \ref{fig:eventualities}).

\begin{definition}
Given a model $\model$, an \emph{eventuality} in $\model$ is a pair $(w, \varphi)$, where $w \in W$ 
and $\varphi$ is a formula such that
either $\varphi = \ubox\psi$ for some formula $\psi$ and $\model, w \nsat \varphi$,
or $\varphi = \psi \Until \chi$ for some formulas $\psi$ and $\chi$ and $\model, w \sat \varphi$.
The \emph{fulfillment} of an eventuality $(w, \varphi)$ is
the finite sequence $v_0 \ldots v_n$ of states of the model such that
\begin{enumerate}
  \item for all $k \le n$, $v_k = \msuccfct[^k](w)$,
  \item if $\varphi = \ubox\psi$ then
\begin{enumerate}

\item   $\model, v_n \nsat \psi$ (the {\em end condition for $\varphi$}) and
 \item       for all $k < n$, $\model, v_k \sat \psi$ (the {\em progressive condition for $\varphi$}), and
  \end{enumerate}

  \item if $\varphi = \psi \Until \chi$ then
\begin{enumerate}

\item  $\model, v_n \sat \chi$ (the {\em end condition for $\varphi$}) and
  
\item        for all $k < n$, $\model, v_k \sat \psi$ and $\model, v_k \nsat \chi$ (the {\em progressive condition for $\varphi$}).
        
\end{enumerate}
\end{enumerate}
We call $n$ the {\em fulfillment time} of $(w,\varphi)$. Given a set of formulas $\Sigma$, the fulfillment time of $w$ with respect to $\Sigma$ is the supremum of all fulfillment times of any eventuality $(w,\varphi)$ with $\varphi \in \Sigma$, and if $U $ is a set of worlds or eventualities, the fulfillment time of $U$ with respect to $\Sigma$ is the supremum of all fulfillment times with respect to $\Sigma$ of all elements of $U$.
\end{definition}

The idea is to replace an arbitrary stratified model $\model$ by a related model $\model'$ where all eventualities of $\model'_0$ are realized in effective time.
From such a model $\model'$ we can then extract an effectively bounded finite model $\model\Qfin$.
The model $\M'$ is a `good' model, defined as follows.

\begin{figure}[h!]
\centering
\begin{tikzpicture}[node distance=3.5cm, triangle/.style={ regular polygon, draw, regular polygon sides=3, shape border rotate=180 }] 

  \node (p0) {};
  \node (p1)[right=2cm of p0] {$a$};
  \node (p2)[right=2cm of p1] {$\hdots$};	
  \node (p3)[right=2cm of p2] {$k$};	
  \node (p4)[right=2cm of p3] {$i$};		
  \node (p5)[right=2cm of p4] {$\cdots$};

  \node[triangle,minimum size=3cm] (t1)[above=0cm of p1] {};	
  \node[triangle,minimum size=2cm] (t2)[above=0cm of p2] {};	
  \node[triangle,minimum size=2cm] (t3)[above=0cm of p3] {};	
  \node[triangle,minimum size=2cm] (t4)[above=0cm of p4] {};	

  \node (f1)[above=1.4cm of p1] {\footnotesize $(w,\varphi)$};	
  \node (f2)[above=0.9cm of p1] {\footnotesize $(w',\varphi')$};

  \node (f3)[above=0.8cm of p2] {\footnotesize $\bullet$};	
  \node (f4)[above=0.6cm of p2] {\footnotesize $\bullet$};

  \node (f5)[above=0.8cm of p3] {\footnotesize $\bullet$};	
  \node (f6)[above=0.6cm of p3] {\footnotesize $\bullet$ };	
  \node (varphi)[above=1cm of p3] {\footnotesize $\varphi$};	
  \node (f7)[above=0.6cm of p4] {\footnotesize $\bullet$};	

  \path[successor relation]
  (p0) edge node {} (p1)
  (p1) edge node {} (p2)
  (p2) edge node {} (p3)
  (p3) edge node {} (p4)
  (p4) edge node {} (p5)
  (f1.east) edge node {} (f3)
  (f2.east) edge node {} (f4)	
  (f3.east) edge node {} (f5)
  (f4.east) edge node {} (f6)
  (f6.east) edge node {} (f7)
  ;

\end{tikzpicture}
\caption{The stratum $\M_a$ and two of its eventualities. The fulfillment of $(w,\varphi)$ is displayed, as well as the initial portion of the fulfillment of $(w',\varphi')$.}
\label{fig:eventualities}
\end{figure}
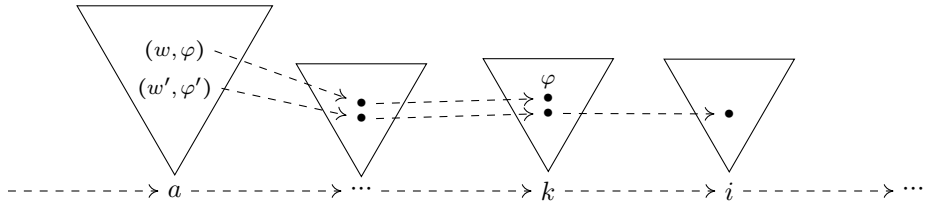

\begin{definition}\label{defGood}
Let $\Sigma \Subset \langred$, $s = |\Sigma|$ and $a,b $ be natural numbers.
An expanding model $\M$ is {\em good (with parameters $a$, $b$, relative to $\Sigma$)} if
\begin{enumerate}

\item \label{itBoundC} $a < b\leq 2 \nre{2^{n+1}}{n + 1} + \nrq{2^n}{n+1} \nre{2^{n+1}}{n + 3}$,

\item $\M_a \simm_\Sigma \M_b$,

\item $W_a$ has fulfillment time less than $b-a$, and

\item for all $c < b $, $\M_c$ is bounded by $\nrq{2^{s+1}}{s+3}$.
\end{enumerate}
\end{definition}

The bound \eqref{itBoundC} will naturally arise throughout our construction, but the only relevance is that it is computable.
We construct $\M'$ as a speedup of $\M$, in a sense that we make precise next.

\begin{definition}\label{defSpeedup}
Let $\Sigma \Subset \langred$, $\model$, $\N$ be stratified models, and $a\leq b'\leq b$ be natural numbers. We say that $\N$ is a {\em speedup of $\model$ from $a$ taking $b$ to $b'$} if for all $i\leq a$ $\N _i = \model _i$ and for all $i\geq b'$ $\N_i = \model_{i+b-b'}$. We say that $\N$ is a {\em strict speedup of $\M$} if $b'< b$. We may omit mention of the parameters if we wish to leave them unspecified, e.g.~$\N$ is a speedup of $\model$ from $a$ if there exist $b,b'$ such that $\N$ is a speedup of $\M$ from $a$ taking $b$ to $b'$.
\end{definition}

Then, the following speedups are defined
for any stratified model $\modelbase$ and any finite, non-empty set of formulas $\Sigma$ closed
under subformulas. In each case, if $\model=(W,{\peq},\msuccfct,V)$ is a stratified model, we will
produce another stratified model $\model'=(W',{\peq}',\msuccfct',V')$ and a map $\pi\colon W'\to W$
such that $\Sigma_\model(\pi(w))=\Sigma_{\model'}(w)$ for all $w\in W'$.
Below, recall that $\model_k = (W_k,{\peq}_k,\msuccfct_k,V_k)$ denotes the $k^{\rm th}$ stratum of $\model$.

\begin{enumerate}[label=({\sc su}\arabic*)]

\item\label{transOne} {\em Replace $\model_k$ with a copy of the normalized $\Sigma$-quasimodel of $\model_k$,} where $k\geq 0$.
Let $\T = \left(W_\T, \isuccrel_\T, \labfun_\T\right)$ be a copy of the normalized labelled tree of $\model^\Sigma_k$
such that $W_\T \cap W = \varnothing$,
and $(\srone,\srtwo)$ the condensation from $\model^\Sigma_k$ to $\T$.
The result of the transformation is the tuple $(W', \irel', \msuccfct', V')$ such that
$W' = ( W \setminus W_k ) \cup W_\T $, $\irel' = \reduc{\irel}{W \setminus W_k} \cup \left(\isuccrel_\T\right)^*$,
\[
  \msuccfct'(w) = \begin{cases}
    \srone\left(\msuccfct\left(w\right)\right) &\text{if } w \in W_{k-1} \\
    \msuccfct\left(\srtwo\left(w\right)\right) &\text{if } w \in W_\T \\
    \msuccfct(w) &\text{otherwise}
  \end{cases} \qquad
  V'(w) = \begin{cases}
     \labfun_T(w) \cap \Var &\text{if } w \in W_\T \\
    V(w) &\text{otherwise}
  \end{cases}
\]
The map $\pi$ is the identity on $W'_i=W_i$ for $i\not=k$, and $\pi(w)=\srtwo(w)$ for $w\in W_\T$.

\item\label{transTwo} {\em Replace $(\model_k,w)$ with a copy of its normalized, pointed $ \Sigma$-quasimodel,}
where $k\geq 0$ and $w \in W_k$.
The transformation is similar to the previous one except that
$(\model_k,w)$ is regarded as a pointed structure with distinguished point $w$.

\item\label{transThree} {\em Replace $\model_\ell$ with $\model_k$,}
where $k < \ell$ and there is an immersion $\simvar\colon W_k\to W_\ell$ (seen as $\parts\Sigma$-labelled trees).
The result of the transformation is the tuple $(W', \irel', \msuccfct', V')$ such that $W' = W \setminus \bigcup_{k < m \le \ell} W_m,$ $\irel' = \reduc{\irel}{W'}$,
\[\msuccfct'(w) =
    \begin{cases}
      \msuccfct\left(\simvar\left(w\right)\right) &\text{if } \msuccfct(w) \in W_k \\
      \msuccfct(w) &\text{otherwise}
    \end{cases}\]
and $ V' = \reduc{V}{W'}$.

The map $\pi$ is the identity on $W'_i = W_i$ for $i < k$, on $W'_i = W_{i+\ell-k}$ for $i > k$,
and $\pi(w)=\simvar(w)$ for all $w \in W'_k$.

\item\label{transFour} {\em Replace $(\model_\ell,w_\ell)$ with $(\model_k,w_k)$,}
where $k < \ell$, $w_k \in W_k$, $w_\ell \in W_\ell$ and
there is an immersion $\simvar\colon W_k \to W_\ell$ such that $\simvar(w_k) = w_\ell$.
The transformation is defined as the previous one.

\end{enumerate}

\begin{lemma} \label{lem:transformations}
Let $a < k < \ell \leq b$ be natural numbers and suppose that $\model$ is such that one of the transformations \ref{transOne}-\ref{transFour} applies.
Then, the result $\model'$ is a speedup of $\M$ between $a$ and $b$ such that
  $\Sigma_{\model}({\pi(w)}) = \Sigma_{\model'}(w)$ for any $w\in W'$.
  In the cases \ref{transThree} and \ref{transFour}, the speedup is strict.
  \philippe{I cannot understand the proof of Lemma 6.1. In fact, I cannot at all understand this construction up to the end of Section 6.}
\end{lemma}

\begin{proof}
  The proof that $\modelbase[']$ is a speedup of $\modelbase$ consists of checking that Definition \ref{defSpeedup} applies and is left to the reader.
  We prove by structural induction on $\varphi$ that
  for all transformations,
  all $w \in W'$ and all $\varphi \in \Sigma$,
  $\model', w \sat \varphi$ iff $\model, \pi(w) \sat \varphi$.
  
  We only detail the case for $\varphi = \tnext \psi$ in the sub-case when $\model_k$ is replaced with a copy of the
  normalized $\Sigma$-quasimodel $\T$ of $\model_k$
  and $w \in W'_{k-1}$.
  Suppose that $w \in W'_{k-1}$ and $\model, \pi(w) \sat \tnext \psi$.
  Then $\psi \in \Sigma _\model  (\msuccfct \pi (w) )$.
  Since $\msuccfct'(w) = \srone \msuccfct \pi(w)  $, $\pi \msuccfct'(w)  = \srtwo \msuccfct'(w) $
  and $(\srone,\srtwo)$ is a condensation,
  $\Sigma_\model (\msuccfct \pi (w)) = \labfun_\T(\msuccfct'(w)) = \Sigma_\model (\pi \msuccfct'(w) )$.
  In particular $\psi \in \Sigma_\model (\pi \msuccfct'(w) )$, so that $\model, \pi \msuccfct'(w)  \sat \psi$.
  By induction hypothesis, $\model', \msuccfct'(w) \sat \psi$. Hence $\model', w \sat \tnext \psi$.
  The other direction is similar.

The remaining two sub-cases for $\varphi = \tnext \psi$ are when $w\in W'_k$ and when $w\not \in W'_{k-1} \cup W'_k$, both of which are treated similarly.
    The cases for the other temporal modalities also follow from similar considerations
  (see also the proof of Lemma~\ref{LemFinSigma}).
  The cases for the implication are similar to those in the proof of Proposition~\ref{prop:IL}, and the remaining cases are straightforward.
We leave the details to the reader.\david{Maybe do the more complicated $\psi \Until \theta$ case instead?}
\end{proof}

The purpose of the transformations \ref{transTwo} and \ref{transFour} is to preserve fulfillments of formulas. We make this precise in the next lemma.

\begin{lemma}\label{lemmFulfPres}
Let $\Sigma\Subset \lang_{\ubox\Until}$, $\M = \modelbase$ be a stratified model, $a , k \in \mathbb N$ with $k>0$, and $w\in W_a$.
Suppose that $\varphi \in \Sigma$ is such that $(w,\varphi)$ is an eventuality of $\M$ with fulfillment $w=w_0,\ldots,w_n$.
\begin{enumerate}

\item If $k \leq n$ and $\M'$ is obtained by replacing $(\M_{a+k},w_{k})$ by $(\T,v)$, then $(w,\varphi)$ is an eventuality of $\M'$ and the fulfillment of $(w,\varphi)$ is $v_0,\ldots,v_n$ with $v_k = v$ and otherwise $v_i = w_i$.

\item If $k < \ell \leq  n$ and $\M'$ is obtained by replacing $(\M_{a+\ell},w_{\ell})$ by $(\M_{a+k},w_{k})$, then $(w,\varphi)$ is an eventuality of $\M'$ and the fulfillment of $(w,\varphi)$ is $w_0,\ldots,w_k,w_{\ell+1},\ldots w_n$.

\end{enumerate}
\end{lemma}

The proof is straightforward and left to the reader.
In the next few lemmas we show that models can always be sped up so that fulfillment times are effectively bounded.

\begin{lemma}\label{lemSpeedAB}
Fix $\Sigma \Subset \langred$ with $s = |\Sigma|$ and let $\M$ be any startified model and $a<b$ be natural numbers. Then there is a speedup of $\M'$ of $\M$ from $a$ taking $b$ to some $b' \leq a + \nre{2^s}{s+1}$, and such that $\M'_i $ is bounded by $\nrq{2^{s}}{s+1}$ for all $i\in (a,b')$.
\end{lemma}

\proof
Let $b'$ be minimal such that some model $\N$ is a speedup of $\M$ from $a$ taking $ b$ to $b'$.
We claim that $b' \leq a + \nre{2^s}{s+1}$; for otherwise, by Theorem \ref{TheoKruskal}.\ref{ItKruskal} there are natural numbers $i,j$ with $a < i < j \leq b'$ such that $\N_i \simm_\Sigma \N_j$, and hence we can apply a transformation \ref{transThree} to obtain some speedup $\N'$ of $\N$ from $a$ taking $b'$ to some $b''<b'$; but then clearly $\N'$ is also a speedup of $\M$ from $a$ taking $b$ to $b''$ and $b''<b'$, a contradiction.

Thus $b' \leq a + \nre{2^{s}}{s+1}$, and finally we obtain $\M'$ by replacing each $\N_x$ with $x\in (a,b')$ by its normalized $\Sigma$-quasimodel, which by Proposition \ref{PropBound} is bounded by $\nrq{2^{s}}{s+1}$.
\endproof

\begin{lemma}\label{lemSpeedFulf}
Fix a finite set $\Sigma\Subset \langred$ with $s = |\Sigma|$ and let $\M = (W,\peq,S,V)$ be any stratified model, $a\in \mathbb N$, and $U \subseteq W_a \times \Sigma$ be a finite set of eventualities. Then there is a speedup $\N$ of $\M$ from $a$ such that the fulfillment time $\ell$ of $U$ in $\N$ satisfies
\begin{enumerate}

\item $ \ell \leq |U|\nre{2^{s+1}}{s + 3}$, and 

\item for all $x\in [1, \ell -a)$, $\N_{a+x}$ is bounded by $\nrq{2^{s+1}}{s + 3}$.

\end{enumerate}
\end{lemma}

\proof
By induction on $|U|$. The claim is vacuously true if $U =\varnothing$. Otherwise, let $n+1 =|U|$ and $(w,\varphi) \in U$ and assume inductively that a speedup $\M'$ of $\M$ from $a$ is given so that the fulfilment time of $U \setminus \{(w,\varphi)\}$ in $\M'$ is $\ell \leq n \nre{2^{s+1}}{s + 3}$ and for all $x< \ell -a$, $\N_{a+1+x}$ is bounded by $\nrq{2^{s+1}}{s + 3}$.

Let $\N$ be a speedup of $\M'$ from $a+\ell$ chosen so that the fulfilment time $r$ of $(w,\varphi)$ in $\N$ is least among all such speedups.
We claim that $r \leq (n+1) \nre{2^{s+1}}{s + 3}$.
If not, let $w_0,\ldots,w_r$ be the fulfilment path for $(w,\varphi)$, and for $x \in [ 1 ,r - \ell ]$ let $\N_{\ell + x}^+$ be the pointed submodel $(\N_{\ell + x},w_{\ell+x})$.
Note that $r-\ell > \nre{2^{s+1}}{s + 3}$, so that by Proposition \ref{PropBoundPointed} there are $x , y\in \mathbb N$ such that $0< x < y \leq r-\ell$ and $\N_{\ell + x}^+ \simm_\Sigma \N_{\ell + y}^+$. Thus we can apply a transformation \ref{transFour} and replace $\N^+_{\ell + y}$ by $\N^+_{\ell + x}$ to obtain a speedup $\N'$ of $\N$. By Lemma \ref{lemmFulfPres}, the fulfilment of $(w,\varphi)$ in $\N'$ is $w_0,\ldots,w_{\ell+x},w_{\ell+y+1},\ldots ,w_r$, so that $(w,\varphi)$ has fulfilment time $r-(y-x)$, contradicting the minimality of $r$.

Finally we define $\N'$ by replacing each $(\N_{\ell + x},w_{\ell+x})$ with $x\in [0,r-\ell)$ by its pointed, normalized $\Sigma$-quasimodel, which in view of Proposition \ref{PropBoundPointed} has size at most $\nrq{2^{s+1}}{s + 3}$ and by Lemma \ref{lemmFulfPres} preserves the fulfilment time of $(w,\varphi)$, as needed.
\endproof

In the next lemmas we construct a good model in three phases, each time obtaining more of the properties required by Definition \ref{defGood}.
Below, if $\Sigma \Subset \langred$ and $\M =(W,\peq,S,V)$ is a stratified model and $a\in \mathbb N$, we say that $\M_a$ {\em occurs infinitely often (with respect to $\Sigma$)} if there are infinitely many values of $i$ such that $\M_a \simm_\Sigma \M_i$.

\begin{lemma}\label{lemmPhaseOne}
Let $\Sigma\Subset \lang_{\ubox\Until}$ and $s = |\Sigma|$ and $\varphi \in \Sigma$. Then $\varphi$ is satisfiable (falsifiable) over the class of expanding posets if and only if $\varphi$ is satisfied in an expanding model $\M$ for which there exists $a \leq \nre{2^s}{s + 1}$ such that
\begin{enumerate}

\item $\M_a$ occurs infinitely often and

\item for all $i\leq a$ the size of $\M_i$ is bounded by $\nrq{2^{s}}{s+1}$.
 
\end{enumerate}
\end{lemma}

\proof
Suppose that $\varphi$ is satisfiable (falsifiable).
Then, by Theorem \ref{TheoStrat}, $\varphi$ is satisfied (falsified) on $\N_0$ for some stratified model $\N$.
By Proposition \ref{PropBound} there are finitely many $\simm_\Sigma$ equivalence classes, and hence there is some $a'$ such that $\N_{a'}$ occurs infinitely often.

By Lemma \ref{lemSpeedAB} there is a speedup $\M'$ of $\N$ from $0$ taking $a'$ to some $a \leq \nre{2^s}{s+1}$ and such that the size of $\M'_i$ is bounded by $\nrq{2^{s}}{s+1}$ for all $ i \in (0,a)$. It is then easy to see that $\M'_a$ occurs infinitely often in $\M'$.
Finally we define $\M$ by replacing $\M'_0$ and $\M'_a$ by their normalized $\parts\Sigma$-labelled trees, which by Proposition \ref{PropBound} have size at most $\nrq{2^{s}}{s+1}$.
\endproof

\begin{lemma}\label{lemmPhaseTwo}
Let $\Sigma\Subset \langred$ with $s = |\Sigma|$ and $\varphi \in \Sigma$. Then $\varphi$ is satisfiable (falsifiable) over the class of dynamic posets if and only if $\varphi$ is satisfied in a stratified model $\M$ for which there exists $a \leq \nre{2^s}{s+1}$ such that
\begin{enumerate}

\item $\M_a$ occurs infinitely often,

\item $W_a$ has fulfilment time $r\leq s \nrq{2^s}{s+1} \nre{2^{s+1}}{s + 3}$, and 

\item for all $i \leq a + r $, $\M_i$ is bounded by $\nrq{2^{s+1}}{s+3}$.
\end{enumerate}
\end{lemma}

\proof
In view of Lemma \ref{lemmPhaseOne}, we may assume that $\varphi$ is satisfied (falsified) on $\N_0$ for some expanding model $\N = (W,\peq,S,V)$ satisfying the first condition and such that for all $ i \leq a$ the size of $\M_i$ is bounded by $\nrq{2^{s}}{s+1} $.
Let $U \subseteq W_a\times \Sigma$ be the set of all eventualities of $\N_a$; by Lemma \ref{lemSpeedFulf} there is a speedup $\M'$ of $\N$ from $a$ such that the realization time of $\M'_a$ is bounded by $| U | \nre{2^{s+1}}{s + 3}$ and such that $\M'_{a+1+i}$ is bounded by $\nrq{2^{s+1}}{s+3}$ for all $ i < r$.
Clearly $|U| \leq s|W_a| \leq s \nrq{2^s}{s+1}  $, giving us the second condition.
Since for $i\leq a$ we have that $\M_i$ is bounded by $ \nrq{2^{s}}{s+1} \leq \nrq{2^{s+1}}{s+3} $, we obtain the third condition.
\endproof

Finally we are able to show that satisfiability and validity can be restricted to good models.

\begin{lemma}\label{lemmPhaseThree}
Let $\Sigma\Subset \langred$ with $s = |\Sigma|$ and $\varphi \in \Sigma$. Then $\varphi$ is satisfiable (falsifiable) over the class of expanding posets if and only if $\varphi$ is satisfied (falsified) in a good model. 
\end{lemma}

\proof
We may begin with a model $\N =(W,\peq,S,V)$ satisfying all conditions of Lemma \ref{lemmPhaseTwo}, where $\N_a$ occurs infinitely often and $r$ is the realization time of $W_a$. Since $\M_a$ occurs infinitely often, we may choose $b'>a+r$ such that $\N_a \simm_\Sigma \N_{b'}$. Then, by Lemma \ref{lemSpeedAB} there is a speedup $\M$ of $\N$ from $a+r$ taking $b'$ to some $b \leq a + r + \nre{2^{s+1}}{s + 1}$ and such that $\M_i$ is bounded by $\nrq{2^{s+1}}{s+1}$ (and hence by $\nrq{2^{s+1}}{s+3}$) for all $ i \in (a+r,b)$. The model $\M$ then has all desired properties.
\endproof

\begin{definition}
Let $\model$ be an expanding model such that there is an immersion $\simvar \colon W_b  \to W_a$.
Then we define a new pointed model $\model\Qfin = (W\Qfin, \irel\Qfin, \msuccfct\Qfin, V\Qfin,w_0\Qfin)$ by setting $W\Qfin = \bigcup_{0 \le m < b} W_m $, $\irel\Qfin = \reduc{\irel}{W\Qfin}$,
\[
  \msuccfct\Qfin(w) = \begin{cases}
    \simvar\left(\msuccfct\left(w\right)\right) &\text{if } w \in W_{b-1} \\
    \msuccfct(w) &\text{otherwise}
  \end{cases}
  \]
$ V\Qfin = \reduc{V}{W\Qfin}$, and $w_0\Qfin$ to be the root of $W_0$ (note that $w_0\Qfin \in W\Qfin$).
\end{definition}

The idea is to apply the operation $\cdot\Qfin$ to good models, in which case the end result is a well-behaved finite model as described in the next lemma and Figure \ref{fig:final-model}.

\begin{figure}[h]
  \begin{center}
	\begin{tikzpicture}[node distance=3cm, triangle/.style={ regular polygon,draw, regular polygon sides=3 }] 

	\node (p0) { };
	\node (pp1)[right=1.2cm of p0]{ };
	\node (pp2)[right=1.2cm of pp1]{ };		
	\node[triangle,shape border rotate=180,minimum size=1cm] (tt1)[above=-.1cm of pp1] {};	
	\node[triangle,shape border rotate=180,minimum size=1cm] (tt2)[above=-.1cm of pp2] {};			
	\node (p1)[right=1.2cm of pp2] { };
	\node (p3)[right=1.2cm of p1] { };	
	\node (p4)[right=1.2cm of p3] { };		
	
	\node (pp3)[right=1.2cm of p4]{ };
	\node (pp4)[right=1.2cm of pp3]{ };		
	\node (pp5)[right=1.2cm of pp4]{ };
	
	\node (l1)[below=.4cm of p0] {$0$};	
	\node (l2)[below=.4cm of p1] {$a$};	
	\node (l3)[below=.4cm of pp3] {$a + r$};		
	\node (l4)[below=.4cm of pp5] {$b-1$};

	\node[triangle,shape border rotate=180,minimum size=1cm] (t1)[above=-.1cm of p0] {};	
	\node[triangle,shape border rotate=180,minimum size=2cm] (t1)[above=-.1cm of p1] {};	
	\node[triangle,shape border rotate=180,minimum size=1cm] (t3)[above=-.1cm of p3] {};	
	\node[triangle,shape border rotate=180,minimum size=1cm] (t4)[above=-.1cm of p4] {};		
	\node[triangle,shape border rotate=180,minimum size=1cm] (tt3)[above=-.1cm of pp3] {};	
	\node[triangle,shape border rotate=180,minimum size=1cm] (tt4)[above=-.1cm of pp4] {};	
	\node[triangle,shape border rotate=180,minimum size=1cm] (tt5)[above=-.1cm of pp5] {};

	\path[successor relation]
	(p0)  edge node {} (pp1)
	(pp1) edge node {} (pp2)
	(pp2) edge node {} (p1)
	(p1)  edge node {} (p3)
	(p3)  edge node {} (p4)
	(p4)  edge node {} (pp3)
	(pp3) edge node {} (pp4)
	(pp4) edge node {} (pp5)
	(pp5) edge[bend right=50,out=-55] node {} (p1)
  ;
	
	\path[<->,>=stealth']
	(l1) edge[] node[above] {Phase 1}
	            node[below] { \small $\nre{2^s}{s+1}$} (l2)
	(l2) edge[] node[above] {Phase 2}
	            node[below] {\small $s \nrq{2^s}{s+1} \nre{2^{s+1}}{s + 3}$} (l3)
	(l3) edge[] node[above] {Phase 3}
	            node[below] {\small $\nre{2^s}{s+1}$} (l4)
	;

	\end{tikzpicture}
	\end{center}
	\caption{An illustration of the three phases of $\model\Qfin$ built from a good model. Below each phase we indicate the maximum number of strata, used for the computations in the proof of Lemma \ref{lemFinalBound}.}
	\label{fig:final-model}
\end{figure}
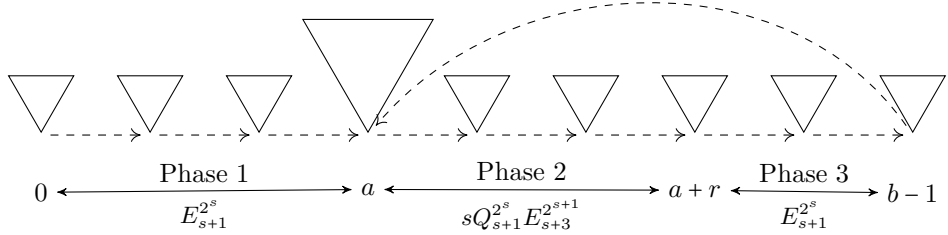\david{The lines in this figure are not perfectly horizontal, we should fix that.}

\begin{lemma}\label{LemFinSigma}
If $\model = (W,\peq,S,V)$ is a good model with parameters $a$, $b$ then $\model\Qfin$ is a model and $\Sigma_{\model\Qfin} ( w ) = \Sigma_\model ( w ) $ for all $w\in W\Qfin$.
\end{lemma}

\begin{proof}
  The proof that $\modelbase[\Qfin]$ is a model is straightforward and left to the reader.
  We prove by structural induction on $\varphi$ that
  for all $w \in W\Qfin$ and all $\varphi \in \Sigma$,
  $\model\Qfin, w \sat \varphi$ iff $\model, w \sat \varphi$.
  The cases for propositional variables and the Boolean connectives are straightforward.
  The case for the `next' temporal modality is similar to that in the proof of
  Lemma~\ref{lem:transformations}.

  For the `henceforth' and `until' temporal modalities,
  suppose first that $(w, \varphi)$ is an eventuality in $\model$ and $w \in W\Qfin$.
  Let $w_0 \ldots w_n$ be the fulfilment of $(w, \varphi)$ in~$\model$.
  If $w_n \in W\Qfin$ then we can apply the induction hypothesis to see that each $w_i$ for $i\leq n$ satisfies the progressive and the end conditions for $(w,\varphi)$ in $\model\Qfin$: if $\varphi = \theta \Until \psi$ then $\M,w_n\models \psi$ and for all $i<n$ $\M,w_i \models \theta$ and $\M,w_i \not \models \psi$, which by induction on formula length yields $\M\Qfin,w_n\models \psi$ and for all $i<n$ $\M\Qfin ,w_i \models \theta$.
  The case for $\varphi = \ubox \psi$ is similar.
  
  Otherwise, there is a least $k \le n$ such that $w_k \in W_b$.
  Therefore, $(w_k, \varphi)$ is an eventuality in $\model$ and so is $(\sigma(w_k), \varphi)$ since $\sigma$ is an immersion.
  Since $\model$ is good, 
  the length of the fulfilment of any eventuality $(v, \varphi)$ such that $v \in W_a$
  is bounded by $  b - a$.
  Thus by the previous case (where $w_n \in W\Qfin$), $(\sigma(w_k), \varphi)$ is an eventuality in $\model\Qfin$.
  Let $v_0,\ldots,v_\ell$ be its fulfilment.
  Then it is not hard to see using the induction hypothesis that $w_0,\ldots,w_{k-1},v_0,\ldots,v_\ell$ is the fulfilment of $(w,\varphi)$ in $\model\Qfin$, witnessing that $\model\Qfin,w\models \varphi$.
 
  Conversely, suppose now that $(w, \varphi)$ is an eventuality in $\model\Qfin$ and
  let $w_0 \ldots w_n$ be its fulfilment.
  For each $k \le n$ let $m_k$ be such that $w_k \in W_{m_k}$.
  The proof is by a subinduction on $n$. For the base case we directly apply the induction hypothesis to $w = w_n$.
If $n>0$ then first note that by the main induction hypothesis on $\varphi$, the sequence $w_0 \ldots w_n$ satisfies the progressive condition for $(w,\varphi)$ on $\model$.

Now consider two cases. If $m_0<b-1$ then $m_{1} < b$. The sub-induction hypothesis tells us that $(w_1 ,\varphi)$ is an eventuality of $\M$, and since $\M,w_0$ satisfies the progressive condition for $(w,\varphi)$ it follows that $(w,\varphi)$ is an eventuality of $\M$.

Otherwise $m_0 = b -1$, so that $m_1 = a$.
The sub-induction hypothesis tells us that $(w_1,\varphi)$ is an eventuality of $\M$.
  Since $w_1 = S\Qfin(w_{0}) = \sigma S(w_{0}) $ and
  $\sigma$ is an immersion, $(S(w_{0}), \varphi)$ is an eventuality in $\model$.
  Therefore, $(w, \varphi)$ is an eventuality in $\model$.
\end{proof}

\begin{lemma}\label{lemFinalBound}
If $\M$ is a good model with parameters $a,b$ and $s = \card\Sigma$ then $\M\Qfin$ is bounded by
\[
    B(s) \eqdef \nrq{2^{s+1}}{s+3}\left(2 \nre{2^s}{s+1} + s \nrq{2^s}{s+1} \nre{2^{s+1}}{s + 3}\right)
  \]
\end{lemma}

\proof
This is immediate from the definition of $W\Qfin$ and the bounds on good models (see Defininition \ref{defGood}).
\endproof

We have proven the following effective finite model property for $\langred$; however, since $\lang$ maps effectively into $\langred$, this result applies to the full language.

\begin{theorem} \label{ThmFMP}
  There exists a computable function $B$ such that
  for any formula $\varphi \in \lang$,
  if $\varphi$ is satisfiable (resp. unsatisfiable)
  then $\varphi$ is satisfiable (resp. falsifiable) in a model $\modelbase$
  such that $\card W \le B(\nos\varphi)$.
\end{theorem}

\fullproof{
\begin{proof}
  In view of Theorem~\ref{TheoStrat},
  a formula $\varphi$ is satisfiable (resp. falsifiable) in a model $\model$ if and only if
  it is satisfied (resp. falsified) at the root of a stratified model $\model\Qstrat$.
  Then, by Lemma~\ref{LemFinSigma}, $\varphi$ is satisfied (resp. falsified) in $\model\Qstrat$ if and only if
  it is satisfied (res. falsified) on $(\model\Qstrat)\Qfin$,
  which is effectively bounded by $B(\nos\varphi)$ by Lemma \ref{LemmFinalBound}.
\end{proof}
}
\shortproof{}

As a corollary, we get the decidability of $\iltl$.

\begin{corollary}\label{CorITLeDec}
  The satisfiability and validity problems for $\iltl$ are decidable.
\end{corollary}

\section{Bounded bisimulations for $\Until$ and $\R$}\label{sec:bisim}

\david{Double-check definitions.}

In this section we adapt the classical definition of bounded bisimulations for modal logic~\cite{BRV01} to our case. To do so we combine the ordinary definition of bounded bisimulations with the work of~\cite{P97} on bisimulations for propositional intuitionistic logic, which includes extra conditions involving the partial order $\irel$. In our setting, we combine both approaches in order to define bisimulation for a language involving $\imp$, $\tnext$, $\Until$ and $\R$, where the latter are adapted from bisimulations for  a language with \emph{until} and \emph{since}~\cite{Kamp68} presented by Kurtonina and de Rijke~\cite{KR97}.
Since all languages we consider contain Booleans and $\tnext$, it is convenient to begin with a `basic' notion of bisimulation for this language.

\begin{definition}
Given $n > 0 $ and two $\iltl$ models $\model_1$ and $\model_2$, a sequence of binary relations $\bisim_n \subseteq \cdots \subseteq \bisim_0 \subseteq W_1 \times W_2$ is said to be a 
\emph{bounded $\tnext$-bisimulation} if for all $(w_1,w_2)\in W_1 \times W_2$ and for all $0\le i < n$, the following conditions are satisfied:
\smallskip

\noindent{\sc Atoms.} If $w_1 \bisim_i w_2$ then for all propositional variables $p$, $\model_1,w_1 \sat p$ iff $\model_2, w_2\sat p$.\smallskip

\noindent{\sc Forth $\to$.} If $w_1 \bisim_{i+1} w_2$ then for all $v_1 \in W_1$, if $v_1 \succcurlyeq w_1$, there exists $v_2\in W_2$ such that $v_2 \succcurlyeq w_2$ and $v_1 \bisim_i v_2$.
\smallskip

\noindent{\sc Back $\to$.} If $w_1 \bisim_{i+1} w_2$ then for all $v_2 \in W_2$ if $v_2\succcurlyeq w_2$ then there exists $v_1\in W_1$ such that $v_1 \succcurlyeq w_1$ and $v_1 \bisim_iv_2$.
\smallskip

\noindent{\sc Forth $\tnext$.} if $w_1 \bisim_{i+1} w_2$ then $\msuccfct(w_1) \bisim_i \msuccfct(w_2)$.
\end{definition}

Note that there is not `back' clause for $\tnext$; this is simply because $S$ is a function, so its `forth' and `back' clauses are identical. Bounded $\tnext$-bisimulations are useful because they preserve the truth of relatively small $\lang_{\tnext}$-formulas.

\begin{lemma}\label{lemma:bisimulation:tnext}
Given two $\iltl$ models $\model_1$ and $\model_2$ and a bounded $\tnext$-bisimulation $\bisim_n \subseteq \cdots \subseteq \bisim_0$ between them, for all $i\leq n$ and $(w_1,w_2)\in W_1\times W_2$, if $w_1\bisim_i w_2$ then for all $\varphi \in \lang_{\tnext}$ satisfying\footnote{Although not optimal, we use the length of the formula in this lemma to simplify its proof. More precise measures like counting the number of modalities and implications could be equally used.} $\length{\varphi} \le i$, $\model_1, w_1 \sat \varphi \hbox{ iff } \model_2, w_2 \sat \varphi$.									
\end{lemma}

\begin{proof}We proceed by induction on $i$. Let $0 \leq i \leq n$ be such that for all $j < i$ the lemma holds. Let $w_1\in W_1$ and $w_2 \in W_2$ be such that $w_1 \bisim_i w_2$ and let us consider $\varphi \in \lang_{\tnext}$ such that $\length{\varphi}\le i$. The cases where $\varphi$ is an atom or of the forms $\theta\wedge \psi$, $\theta\vee \psi$ are as in the classical case and we omit them. Thus we focus on the following:\smallskip

				
\noindent{\sc Case $\varphi = \theta \imp \psi$.} We proceed by contrapositive to prove the left-to-right implication. Note that in this case we must have $i>0$.

Assume that $\model_2, w_2 \not \sat \theta \imp \psi$. Therefore there exists $v_2 \in W_2$ such that $v_2 \succcurlyeq w_2$, $\model_2, v_2 \sat \theta $, and $\model_2, v_2 \not \sat \psi$. By the {\sc Back $\to$} condition, it follows that there exists $v_1\in W_1$ such that $v_1\succcurlyeq w_1$ and $v_1 \bisim_{i-1} v_2$. Since $\length{\theta} , \length{\psi} < i$, by the induction hypothesis, it follows that $\model_1,v_1 \sat \theta$ and $\model_1,v_1 \not \sat \psi$. Consequently, $\model_1, w_1 \not\sat \theta \imp \psi$. The converse direction is proved in a similar way but using {\sc Forth $\imp$}.\smallskip

\noindent {\sc Case $\varphi = \tnext \psi$.} Once again we have that $i>0$. Assume that $\model_1, w_1\sat \tnext \psi$, so that $\model_1, \msuccfct(w_1)\sat \psi$. By {\sc Forth $\tnext$}, $S_1(w_1) \bisim_{i-1} S_2(w_2)$. Moreover, $\length{\psi} \leq i-1$, so that by the induction hypothesis, $\model_2, \msuccfct(w_2)\sat \psi$, and $\model_2, w_2\sat \tnext \psi$. The right-to-left direction is analogous.
\end{proof}	

We will use bounded $\tnext$-bisimulations as a basis to define bounded bisimulations for more powerful languages.
The bisimulations we define below preserve formulas containing the `until' operator.

\begin{definition}
Given $n\in \mathbb{N}$ and two $\iltl$ models $\model_1$ and $\model_2$, a bounded $\tnext$-bisimulation $\bisim_n \subseteq \cdots \subseteq \bisim_0 \subseteq W_1 \times W_2$ is said to be a \emph{bounded $\Until$-bisimulation} iff for all $(w_1,w_2)\in W_1 \times W_2 $ and $0\le i < n$ such that $w_1 \bisim_{i+1} w_2$:\smallskip

\noindent{\sc Forth $\Until$.} For all $k_1\ge 0$ there exist $k_2 \ge 0$ and $(v_1,v_2) \in W_1\times W_2$ such that
\begin{enumerate}[itemsep=0pt]

\item $\msuccfct[^{k_2}](w_2) \succcurlyeq v_2$, $v_1 \succcurlyeq \msuccfct[^{k_1}](w_1)$ and $v_1 \bisim_i v_2$, and

\item  for all $j_2 \in [0,k_2)$  there exist $j_1 \in [0,k_1)$ and $(u_1,u_2) \in W_1 \times W_2$ such that $u_1 \succcurlyeq \msuccfct[^{j_1}](w_1)$, $\msuccfct[^{j_2}](w_2) \succcurlyeq u_2$ and $u_1\bisim_i u_2$.\smallskip

\end{enumerate}

\noindent{\sc Back $\Until$.} For all $k_2\ge 0$ there exist $k_1 \ge 0$ and $(v_1,v_2) \in W_1\times W_2$ such that
\begin{enumerate}[itemsep=0pt]

\item

$\msuccfct[^{k_1}](w_1) \succcurlyeq v_1$, $v_2 \succcurlyeq \msuccfct[^{k_2}](w_2)$ and $v_1 \bisim_i v_2$, and

\item for all $j_1 \in [0,k_1)$  there exist $j_2 \in [0,k_2)$ and $(u_1,u_2) \in W_1 \times W_2$ such that $u_2 \succcurlyeq \msuccfct[^{j_2}](w_2)$, $\msuccfct[^{j_1}](w_1) \succcurlyeq u_1$ and $u_1\bisim_i u_2$.

\end{enumerate}
\end{definition}

As was the case before, the following lemma states that two bounded $\Until$-bisimilar models agree on small-enough $\lang_{\Until}$ formulas.

\begin{lemma}\label{lemma:bisimulation:until}					
	Given two $\iltl$ models $\model_1$ and $\model_2$ and a bounded $\Until$-bisimulation $\bisim_n\subseteq \hdots \subseteq \bisim_0$ between them, for all $m\leq n$ and $(w_1,w_2)\in W_1\times W_2$, if $w_1\bisim_m w_2$ then for all $\varphi \in \lang_{\Until}$ such that $\length{\varphi} \le m$, $\model_1, w_1 \sat \varphi \hbox{ iff } \model_2, w_2 \sat \varphi$.									
\end{lemma}

	\begin{proof} Once again, proceed by induction on $n$. Let $m\leq n$ be such that for all $k < m$ the lemma holds. Let $w_1\in W_1$ and $w_2 \in W_2$ be such that $w_1 \bisim_m w_2$ and let us consider $\varphi \in \lang_{\Until}$ such that $\length{\varphi}\le m$.
	We only consider the new case, where $\varphi = \theta \Until\psi$. From left to right, assume that $\model_1, w_1 \sat \theta \Until \psi$.
	Then, there exists $i_1\ge 0$ such that $\model_1, \msuccfct[^{i_1}](w_1) \sat \psi$ and for all $j_1$ satisfying $0 \le j_1 < i_1$, $\model_1, \msuccfct[^{j_1}](w_1) \sat \theta$.
	By {\sc Forth $\Until$}, there exist $i_2 \ge 0$ and $(v_1,v_2) \in W_1\times W_2$ such that
	\begin{enumerate*}[itemsep=0pt]
		\item $\msuccfct[^{i_2}](w_2)\succcurlyeq v_2$, $v_1 \succcurlyeq \msuccfct[^{i_1}](w_1)$ and $v_1 \bisim_{m-1} v_2$;
		\item\label{itTwoUBis} for all $j_2$ satisfying $0 \le j_2 < i_2$  there exist $j_1 \in [0,i_1)$ and $(u_1,u_2) \in W_1 \times W_2$ s. t. $u_1 \succcurlyeq \msuccfct[^{j_1}](w_1)$, $\msuccfct[^{j_2}](w_2) \succcurlyeq u_2$ and $u_1\bisim_{m-1} u_2$.
	\end{enumerate*}
		
Since $v_1 \succcurlyeq \msuccfct[^{i_1}](w_1)$ and $\model_1, \msuccfct[^{i_1}](w_1) \sat \psi$, by $\peq$-monotonicity we see that  $\model_1, v_1 \sat \psi$. Since $\length{\psi} \leq m-1$, it follows from the induction hypothesis that $\model_2, v_2 \sat \psi$, and by $\peq$-monotonicity, $\model_2, \msuccfct[^{i_2}](w_2) \sat \psi$.

Now take any $j_2$ satisfying $0 \le j_2 < i_2$. Using \eqref{itTwoUBis}, the fact that $\length{\theta} \leq m-1$, and the induction hypothesis, we may reason as above to conclude that $\model_2, \msuccfct[^{j_2}](w_2) \sat \theta$ so $\model_2, w_2 \sat \theta \Until\psi$. The right-to-left direction is symmetric (but uses {\sc Back $\Until$}).
	\end{proof}	

Finally, we define bounded bisimulations for `release'. The idea is similar as that for the `until' operator.

\begin{definition}
A bounded $\tnext$-bisimulation $\bisim_n \subseteq \cdots \subseteq \bisim_0 \subseteq W_1 \times W_2$ is said to be a \emph{bounded $\R$-bisimulation} if for all $(w_1,w_2)\in W_1 \times W_2 $ and $0\le i < n$ such that $w_1 \bisim_{i+1} w_2$:\smallskip

\noindent{\sc Forth $\R$.} For all $k_2\ge 0$ there exist $k_1 \ge 0$ and $(v_1,v_2) \in W_1\times W_2$ such that

\begin{enumerate}[itemsep=0pt]

\item $\msuccfct[^{k_2}](w_2)\succcurlyeq v_2$, $v_1 \succcurlyeq \msuccfct[^{k_1}](w_1)$ and $v_1 \bisim_i v_2$, and

\item for all $j_1$ satisfying $0\le j_1 < k_1$ there exist $j_2$ such that $0 \le j_2 < k_2$ and $(u_1,u_2) \in W_1 \times W_2$ s. t. $u_1 \succcurlyeq \msuccfct[^{j_1}](w_1)$, $\msuccfct[^{j_2}](w_2)\succcurlyeq u_2$ and $u_1\bisim_i u_2$.

\end{enumerate}

\noindent{\sc Back $\R$.} For all $k_1\ge 0$ there exist $k_2 \ge 0$ and $(v_1,v_2) \in W_1\times W_2$ such that

\begin{enumerate}[itemsep=0pt]

\item

$\msuccfct[^{k_1}](w_1)\succcurlyeq v_1$, $v_2 \succcurlyeq \msuccfct[^{k_2}](w_2)$ and $v_1 \bisim_i v_2$, and

\item for all $j_2$ satisfying $0 \le j_2 < k_2$ there exist $j_1$ such that $0 \le j_1 < k_1$ and $(u_1,u_2) \in W_1 \times W_2$ s. t. $u_2 \succcurlyeq \msuccfct[^{j_2}](w_2)$, $\msuccfct[^{j_1}](w_1)\succcurlyeq u_1$ and $u_1\bisim_i u_2$. 
\end{enumerate}
\end{definition}

Once again, we obtain a corresponding bisimulation lemma for $\lang_{\R}$.

\begin{lemma}\label{lemma:bisimulation:release}					
Given two $\iltl$ models $\model_1$ and $\model_2$ and a bounded $\R$-bisimulation $\bisim_n \subseteq \cdots \subseteq \bisim_0$ between them, for all $m\leq n$ and $(w_1,w_2)\in W_1\times W_2$, if $w_1\bisim_m w_2$ then for all $\varphi \in \lang_{\Until}$ such that $\length{\varphi} \le m$, $\model_1, w_1 \sat \varphi \hbox{ iff } \model_2, w_2 \sat \varphi$.									
\end{lemma}

\begin{proof}
As before, we proceed by induction on $n$; the critical case where $\varphi =\theta \R \psi$ follows by reasoning similar to that of Lemma \ref{lemma:bisimulation:until}. Details are left to the reader.
	\end{proof}

\section{Definability and undefinability of modal operators} \label{SecSucc}

In this section, we explore the question of when the basic connectives can or cannot be defined in terms of each other.
It is known that, classically, $\diam$ and $\ubox$ are interdefinable, as are $\Until$ and $\R$; we will see that this is not the case intuitionistically. On the other hand, $\Until$ (and hence $\R$) is not definable in terms of $\diam,\ubox$ in the classical setting \cite{Kamp68}, and this result immediately carries over to the intuitionistic setting, as the class of classical $\sf LTL$ models can be seen as the subclass of that of dynamic posets by letting the partial order be the identity.

It is worth noting that interdefinability of modal operators can vary within intermediate logics. For example, $\wedge$, $\vee$ and $\rightarrow$ are basic connectives in propositional intuitionistic logic, but in the intermediate logic of here-and-there~\cite{Hey30}, $\wedge$ is a basic operator~\cite{A+15,BalbianiDieguezJelia} as is $\rightarrow$~\cite{A+15} while $\vee$ is definable in terms of $\rightarrow$ and $\wedge$~\cite{LK41}. In first-order here-and-there~\cite{L+07}, the quantifier $\exists$ is definable in terms of $\forall$ and $\rightarrow$~\cite{Mints10}. In the modal case, Simpson~\cite{Simpson94} shows that modal operators are not interdefinable in the intuitionistic modal logic $\sf IK$ and Balbiani and Di\'eguez~\cite{BalbianiDieguezJelia} proved that $\ubox$ is not definable in terms of $\diam$ in the linear time temporal extension of here-and-there. This last proof is adapted here to show that $\ubox$ not definable in terms of $\Until$ in $\itlht$ either. Note, however, that here we correct the claim of \cite{BalbianiDieguezJelia} stating that $\diam$ is not here-and-there definable in terms of $\ubox$, although we do show that $\diam$ is not definable in terms of $\R$ over the class of persistent models.

Let us begin by studying the definability of $\ubox$ in terms of $\tnext$ and $\Until$. Recall that $\lang$ denotes the full language of intuitionistic temporal logic. If $\lang'\subseteq \lang$, $\varphi\in \lang$ and $\Omega$ is a class of models, we say that $\varphi$ is {\em $\lang'$-definable over $\Omega$} if there is $\varphi' \in \lang'$ such that $\Omega \models \varphi\leftrightarrow \varphi'$. Thus for example $\diam p$ is $\lang_{\Until}$-definable; however, as we will see, $\ubox p$ is not.

We will show this by exhibiting models that are $n$-$\Until$-bisimilar for arbitrariliy large $n$.
To construct these models, it will be convenient to introduce some ad-hoc notation for cyclic groups. Recall that if $a,b \in \mathbb Z$ we write $a \mid b$ if there is $k\in \mathbb Z$ such that $b=ak$, and $a \equiv b \pmod n$ if $n \mid (a-b)$.
Given $n>0$, we will denote the cyclic group with $n$ elements by $\mathbb Z/(n)$. We will identify it with the set $\{1,\ldots, n\}$, and define $[i]_n$ to be the unique $ j \in [1 , n] $ such that $i \equiv j \pmod n$.
Note that addition in $\mathbb Z/(n)$ is given by $[x + y]_n$.
With this, we are ready to show that $\ubox$ is not definable in terms of $\Until$.

\begin{theorem}
The formula $\ubox p$ is not $\lang_{\Until}$-definable, even over the class of finite here-and-there models.
\end{theorem}

\begin{proof}
For $n>0$ consider a model $ \HTMod n = (W,{\peq},S,V)$ with $W = \big ( \mathbb Z / (n+2) \big )  \times \{0,1\}$, $(i,j) \peq (i',j')$ if $i = i'$ and $j\leq j'$, $S(i,j) = ([i+1]_{n+2} ,j)$, and $V(n+2,0) = \varnothing$, otherwise $V(i,j) = \{p\}$.
Clearly $ \HTMod n$ is a here-and-there model. For $m \leq n$, let $\sim_m$ be the least equivalence relation such that $(i,j) \sim_{m} (i',j')$ whenever
\[\max \{i(1-j) ,i' (1-j') \} \leq n - m +1\]
(see Figure \ref{FigBoxU}). Then, it can easily be checked that $ \HTMod n,(1,0) \not \models \ubox p$, $\model,(1,1) \models \ubox p$, and
$(1,0)\sim_m(1,1)$.

It remains to check that $(\sim_m)_{m\leq n}$ is a bounded $\Until$-bisimulation. The atoms, $\imp$ and $\tnext$ clauses are easily verified, so we focus on those for $\Until$. Since $\sim_m$ is symmetric, we only check {\sc Forth $\Until$}. Suppose that $(i_1,j_1) \sim_m (i_2,j_2)$, and fix $k_1 \geq 0$. Let $i'=[i_1 + k_1]_{n+2}$ and note that  $S^{k_1} (i_1,j_1) = (i ',j_1) $.
Then, we can see that $k_2 = 0$, $v_1 = (i',1)$ and $v_2 = (i_2,j_2)$ witness that {\sc Forth $\Until$} holds, where the intermediate condition for $j_2 \in [0,k_2)$ holds vacuously since $[0,k_2) = \varnothing$.

By letting $n = |\varphi|$, we see using Lemma \ref{lemma:bisimulation:until} that that no $\lang_{ \Until}$-formula $\varphi$ can be equivalent to $\ubox p$.
\end{proof}

\begin{figure}[h!]
\begin{center}
\begin{tikzpicture}[label distance=1pt]

  \foreach \x/\n/\e/\c in {0/0/1/n, 1.5/1/2/n-1, 4/n/n+1/0}
  { 
  	\path (\x,0)   node [p model state, "${(\e,0)}$" below] (n\n0) {};
    \path (\x,1.3) node [p model state, "${(\e,1)}$" above] (n\n1) {};
    \path (\x,0.6) -- +(-.1,0) node [anchor=east,color=black!50!gray] (c\n) {$\c$};
  }
  \path (5.5,0)   node [not p model state, "${(n+2,0)}$" below] (np0) {};
  \path (5.5,1.3) node [    p model state, "${(n+2,1)}$" above] (np1) {};

  \begin{scope}[on background layer]
    \foreach \n/\g/\D in {n/10/7pt, 1/30/5pt, 0/40/3pt}
    \fill[fill=gray!\g,rounded corners]
         ($(n\n0.south east) + (\D,-\D)$)
      -| ($(c0.west) - (\D,0)$) |- ($(np1.north east) + (\D,\D)$)
      |- ($(n\n1.south east) + (\D,-\D)$)
      -- cycle;
    \fill[fill=gray!10,rounded corners] (np0.north west) rectangle (np0.south east);
  \end{scope}

  \foreach \n in {0, 1, n, p}
    \draw[intuitionistic relation] (n\n0) -> (n\n1);
  \foreach \y in {0, 1}
  { \foreach \f/\t in {0/1, n/p}
      \draw[successor relation] (n\f\y) -> (n\t\y);
    \draw[dotted,thick] (n1\y) -- (nn\y); }
  \draw[successor relation,rounded corners] (np0) -| ++(.8,-1) -| (-1,0)   -- (n00);
  \draw[successor relation,rounded corners] (np1) -| ++(.8,1)  -| (-1,1.3) -- (n01);
\end{tikzpicture}

\end{center}
\caption{The here-and-there model $ \HTMod n$. Black dots satisfy the atom $p$, white dots do not; all other atoms are false everywhere.
  Solid lines indicate $\peq$ and dashed lines indicate $S$.
  The $\sim_m$-equivalence classes are shown as grey regions.}
\label{FigBoxU}
\end{figure}
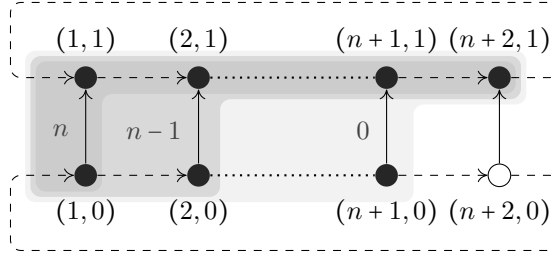

\noindent As a consequence:
\begin{corollary}
The formula $q \R p$ is not definable in terms of $\tnext$ and $\Until$, even over the class of finite here-and-there models.
\end{corollary}	

\begin{proof}
If we could define $q\R p$, then we could also define $\ubox p \equiv \bot \R p$.\end{proof}

The situation is a bit different for $\diam$, at least over the class of here-and-there models.

\begin{proposition}\label{propDiamDefinHT}
Over the class of here-and-there models, $\diam$ is $\lang_{\ubox}$-definable. To be precise, define formulas
\begin{align*}
\alpha& = \ubox (p\to \ubox (p\vee \neg p))\\
\beta &= \ubox (\tnext \ubox (p\vee \neg p) \to p\vee \neg p\vee \tnext \ubox \neg p)\\
\gamma & = \ubox(p\vee \neg p)\wedge \neg \ubox \neg p\\
\varphi& = (\alpha \wedge \beta)\to \gamma.
\end{align*}
Then, $\diam p$ is here-and-there equivalent to $\varphi$.
\end{proposition}

\proof
Let $\model = (T\times \{0,1\},{\peq},S,V)$ be a here-and-there model with $S(t,i) = (f(t),i)$ (see Section \ref{SecHT}). First assume that $x=(x_1,x_2)$ is such that $\model,x \models \diam p$. To check that $\model,x \models\varphi$, let $x'\seq x$, and consider the following cases.\smallskip

\noindent {\sc Case $\model,x' \models \ubox(p\vee \neg p)$.} In this case, it is easy to see that we also have $ \model,x' \models \neg \ubox \neg p$ given that $ \model , x \models \diam p$, so $ \model,x' \models\gamma $. \smallskip

\noindent{\sc Case $ \model,x'  \not \models \ubox(p\vee \neg p)$.} Using the assumption that $ \model,x \models \diam p$, choose $k$ such that $ \model, S^k(x) \models p$ and consider two sub-cases.\\

\begin{enumerate}[itemsep=0pt]

\item\label{CaseA} Suppose there is $k'>k$ such that $ \model, S^{k'}(x) \not \models p\vee \neg p$. Then, it follows that
\[ \model, S^{k}(x') \not \models p \to \ubox ( p\vee \neg p)\]
 and hence $ \model,x' \not \models  \ubox (p\to \ubox (p\vee \neg p)) = \alpha$.

\item\label{CaseB} If there is not such $k'$, then there must be a maximal $k'<k$ such that $ \model, S^{k'}(x' ) \not \models p\vee \neg p$ (otherwise, we would be in {\sc Case $ \model,x' \models \ubox(p\vee \neg p)$}). Since $k'$ is maximal,
\[ \model, S^{k'} (x')  \models  \tnext \ubox (p\vee \neg p),\]
and since $k' < k$ and $ \model, S^{k} (x')  \not \models  \neg p$, we have that
$\model, S^{k'} (x')  \not \models  \tnext \ubox \neg p.$
It follows that
\[ \model, S^{k'} (x' ) \not \models  \tnext \ubox (p\vee \neg p) \to p\vee \neg p\vee \tnext \ubox \neg p ,\]
and therefore
\[ \model, x'  \not \models \ubox (\tnext \ubox (p\vee \neg p) \to p\vee \neg p\vee \tnext \ubox \neg p) = \beta.\]

\end{enumerate}
Since $x'\seq x$ was arbitrary, $\model,x \models (\alpha \wedge \beta)\to \gamma = \varphi$.

Note that the above direction does not use any properties of here-and-there models, and works over arbitrary expanding models. However, we need these properties for the other implication. Suppose that $ \model, x \models \varphi$. If $  \model,x \models \ubox(p\vee \neg p)\wedge \neg \ubox \neg p = \gamma$, then it is readily verified that $ (\model,x) \models \diam p$. Otherwise,
\[ \model,x \not \models \alpha \wedge \beta.\]
If $\model,x \not \models \alpha = \ubox (p\to \ubox (p\vee \neg p))$, then there is $k$ such that
\[\model, S^k(x) \not \models p\to \ubox (p\vee \neg p).\]
Since $ S^k(x) = (f ^k(x_1) ,x_2)$ and $\model, (f ^k(x_1) ,1) \models \ubox (p\vee \neg p)$, this is only possible if $x_2 = 0$ and $ \model,S^k(x ) \models p$, so that $(\model, x)\models \diam p$. Similarly, if
\[ \model,x \not \models \beta = \ubox (\tnext \ubox (p\vee \neg p) \to p\vee \neg p\vee \tnext \ubox \neg p),\]
then there is $k$ such that $ \model, S^k(x) \not \models \tnext \ubox (p\vee \neg p) \to p\vee \neg p\vee \tnext \ubox \neg p$.
Once again using the fact that $ S^k(x) = (f ^k(x_1) ,x_2)$, this is only possible if $x_2 = 0$, $ \model, S^ k (x) \models \tnext \ubox (p\vee \neg p)$ and $ \model, S^k(x) \not \models \tnext \ubox \neg p$. But from this it easily can be seen that there is $k'>k$ with $ \model, S^{k'}(x) \models p$, hence $ \model,x \models \diam p$.
\endproof

\begin{corollary}\label{CorUDef}
Over the class of here-and-there models, $p\Until q$ is $\lang_{\R}$-definable.
\end{corollary}

\begin{proof}
Since $\ubox \varphi$ is definable by $\ubox \varphi \equiv \bot \R \varphi$ and $\diam p$ is definable by Proposition \ref{propDiamDefinHT}, $p \Until q$ is definable by $p\Until q \equiv (q \R (p\vee q))\wedge \diam p$ (Proposition \ref{PropUValid}.\ref{ItUNine}).
\end{proof}

Our goal next is to show that the modality $\diam$ cannot be defined in terms of $\R$ over the class of persistent models.
For this, we will use a model construction based on the last exponent of a number $m > 0$ in base $2$, which we denote by $\ell ( m )$; for example, $6 = 2^2 + 2^1$, so $\ell(6) = 1$.
Before we continue, let us establish some basic properties of the function $\ell$.
The following lemma is easily verified, and we present it without proof.

\begin{lemma}\label{lemEllEqs}
Let $a,b$ be positive integers.
\begin{enumerate}

\item\label{itEllsum} If $\ell(a) < \ell (b)$ then $\ell(a + b) = \ell(a) $ and if $\ell(a) = \ell (b)$ then $\ell(a + b) \geq \ell(a) + 1$.

\item\label{itEllProd} $\ell(ab) = \ell(a) + \ell(b)$.

\item\label{itEllIneq} If $1 \leq a \leq 2^b $ then $\ell(a) \leq b$, and $\ell(a) = b$ if and only if $a = 2^b$.

\end{enumerate}
\end{lemma}

From these properties we obtain the following useful equality.

\begin{lemma}\label{lemInterK}
Let $m\geq 1$, $a\geq 0$ and $k\in [1,2^m)$. Then, $\ell (a 2^m + k) = \ell (k)$.
\end{lemma}

\proof
If $a = 0$, the claim is obvious.
Otherwise, note that since $k < 2^m$, we have that $\ell(k) \leq m-1$.
Then $\ell (a 2^m  ) =  \ell(a) + m \geq m$, so that $\ell ( a 2^m + k  ) = \min \{\ell (a 2^m ), \ell(k) \} = \ell (k)$.
\endproof

With this we are ready to define the models $\DiamMod n$.

\begin{definition}
Let $n\geq 0$ and fix a `designated' variable $p$. We define a model $\DiamMod n =(W,{\peq},S,V)$, where
\begin{enumerate}

\item $W = \mathbb Z / (2^n) \times [0,n ]$,

\item $(i,j) \peq (i',j')$ if $i = i'$ and $j \leq j'$,

\item $S(i,j) = ([i+ 1]_{2^n},j)$, and

\item $V(i,j) = \{p\} $ if and only if $ j  > n - \ell (i)$, $V(i,j) = \varnothing$ otherwise.

\end{enumerate}
\end{definition}

\begin{figure}[h!]
  \begin{tikzpicture}[x=1.2cm,y=1.2cm]
    \foreach \i/\j/\c in {
      1/0/a, 2/0/b, 3/0/c, 4/0/d, 5/0/a, 6/0/b, 7/0/c, 8/0/d,
      1/1/a, 2/1/b, 3/1/c, 4/1/d, 5/1/e, 6/1/f, 7/1/g,
      1/2/e, 2/2/f, 3/2/g,        5/2/e, 6/2/f, 7/2/g,
      1/3/i,        3/3/k,        5/3/i,        7/3/k       }
    { \path (\i,\j) node [not p big model state] (n\i\j) {$\c$} ; }
    \foreach \i/\j/\c in {
      2/3/j, 4/3/l, 6/3/j, 8/3/l,
             4/2/h,        8/2/h,
                           8/1/h}
    { \path (\i,\j) node [p big model state] (n\i\j) {$\c$} ; }
    \foreach \j in {0, 1, 2, 3}
    { \draw[successor relation,rounded corners] (n8\j) -| ++(.4,-.4) -- ++(-7.8,0) |- (n1\j) ;
      \foreach \f/\t in {1/2, 2/3, 3/4, 4/5, 5/6, 6/7, 7/8}
      { \draw[successor relation] (n\f\j) -- (n\t\j) ; }
      \path (.2,\j) node [gray] (l\j) {\j};
      \begin{scope}[on background layer]
      \foreach \f/\t in {1/4, 5/8}
      { \draw[color=gray!40,pattern=north east lines,pattern color=gray!20,rounded corners]
          ($(n\f\j) + (-.3,.3)$) rectangle ($(n\t\j) + (.3,-.3)$);
      }
      \end{scope}
    }
    \foreach \i in {1, 2, 3, 4, 5, 6, 7, 8}
    { \foreach \f/\t in {0/1, 1/2, 2/3}
      { \draw[intuitionistic relation] (n\i\f) -- (n\i\t) ; }
      \path (\i, 3.6) node [gray] (c\i) {\i};
    }
  \end{tikzpicture}
  \caption{The model $\DiamMod 3$.
  Black states satisfy the atom $p$ while the other states do not.
  All other atoms are false everywhere.
  Solid arrows indicate~$\peq$ and dashed arrows indicate~$S$.
  Letters correspond to the equivalence classes w.r.t. $\sim_2$, i.e.,
  if two states $w, x$ are represented by the same letter then $w \sim_2 x$.
  The hashed regions correspond to $2$-blocks.}  \label{fig:undefdiamR}
\end{figure}
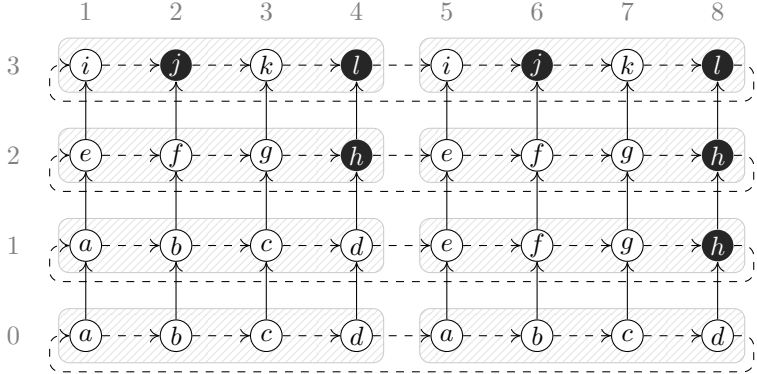

See Figure \ref{fig:undefdiamR} for an illustration of $\DiamMod 3$.
The key properties of the model $\DiamMod n$ are that $(1,0)$ and $(1,1)$ are $(n-1)$-$\R$-bisimilar, yet they disagree on the truth of $\diam p$. Let us begin by proving the latter.

\begin{lemma}\label{lemmDiamDisagree}
Given $n\geq 0$, $\DiamMod n,(1,0)\not \models \diam p$ and $\DiamMod n,(1,1) \models \diam p$.
\end{lemma}

\proof
Let $\DiamMod n=(W,{\peq},S,V)$. Note that $\DiamMod n,(1,1)\models \diam p$ since $(2^n,1) = S^{2^n}(1,0)$ and $1>n - \ell(2^n)$, so that $\DiamMod n,(2^n,1)\models p$. On the other hand, if $(i,j) = S^k(1,0)$ then $j=0$ and $i \in [1,2^n] $, so that by Lemma \ref{lemEllEqs}.\ref{itEllIneq} $\ell(i) \leq n$ and $0 \leq n - \ell(i)$. Hence $\DiamMod n, S^k(1,0) \not \models p$, and since $k$ was arbitrary, $\DiamMod n,(1,0) \not \models \diam p$.
\endproof

Next we will define a family of binary relations $(\sim_m)_{m<n}$ on $\DiamMod n$ which will be used to show that $(1,0)$ and $(1,1)$ are $n$-$\R$-bisimilar. These relations are defined using the notion of {\em congruent blocks.}

\begin{definition}
Let $n\geq 0$ and $\DiamMod n = (W,{\peq},S,V)$.
Given $m \in [0,n]$, say that an {\em $m$-block} is a set of the form
\[B_m(a,b) \eqdef [(a-1) 2^m + 1,a 2^m] \times \{b\},\]
where $a \in \mathbb Z/(2^{n-m})$ and $0\leq b\leq n$; we say that $b$ is the {\em height} of $B_m(a,b)$. Two blocks $B_m(a,b)$ and $B_m(a',b')$ are {\em congruent} if for all $i\in [1, 2^m]$, $ p \in V ((a-1) 2^m+i , b)$ if and only if $p\in V ((a'-1) 2^m+i , b') $. Then, if $x=(x_1,x_2)$ and $y = (y_1,y_2)$, define $x \sim_m y$ if and only if $x_1 \equiv y_1 \pmod{2^m}$ and $x,y$ belong to congruent $m$-blocks.
\end{definition}

It will be convenient to classify the different $m$-blocks.
We say that $B = B_m(a,b)$ is {\em initial} if $V(x) = \varnothing$ for all $x\in B$, {\em terminal} if $V( a 2^m , b) = \{p\}$ and $V(x) = \varnothing$ for any other $x \in B$, and {\em regular} otherwise. A point $x=(x_1, x_2)$ is {\em $m$-initial,} {\em $m$-terminal} or {\em $m$-regular} if it belongs to an $m$-block of the respective kind.
The classification of an $m$-block can be deduced from its height.

\begin{lemma}\label{lemmHeightClass}
Let $n\geq 0$ and $m \in [1,n]$. Let $B = B_m(a,b)$ be an $m$-block in $\DiamMod n$ of height $b$. Then:
\begin{enumerate}

\item\label{itHeightClassInit} $B$ is initial if and only if $b  \leq n - m - \ell(a)$;

\item\label{itHeightClassTerm} $B$ is terminal if and only if $b \in ( n - m - \ell(a) , n - m + 1]$, and

\item\label{itHeightClassReg} $B$ is regular if and only if $b > n - m + 1$.

\end{enumerate}
\end{lemma}

\begin{proof}
Let $B = B_m(a,b)$ be any block. First observe that $\ell (a2^m) = \ell(a) + m \geq m $, so that $p \in V(a2^m,b)$ if and only if $ b > n - \ell (a) - m$; it follows that if $B$ is initial then $b \leq n - \ell (a) - m$.

Next we show that if $b \leq n-m + 1$, then for $k \in [1,2^m)$, $ p \not \in V \big ( (a-1)2^m + k , b' \big )  $.
Since by Lemma \ref{lemInterK} $\ell \big ( (a-1)2^m + k \big ) = \ell (k) < m$, we see that $b \leq n - \ell \big ( (a-1)2^m + k \big )$, and thus $p \not \in V \big ( (a-1)2^m + k , b \big ) $, as claimed.

But then if $b \leq n - \ell (a) - m$ we have that $V  ( a2^m  , b ) = \varnothing  $ as well, so that $B$ is initial if and only if $b \leq n - \ell (a) - m$, while if $b \in ( n - m - \ell(a) , n - m + 1]$ it is neither initial nor regular, hence it is terminal.

It remains to check that if $b > n - m + 1$, then $B$ is regular.
But then as $m\geq 1$ we have that $x = ((a-1)2^m+2^{m-1},b) \in B$ and since $\ell \big ( (a-1)2^m+2^{m-1} \big ) = m - 1 $, we see that $b > n - \ell \big ( (a-1)2^m+2^{m-1} \big )$ and $ x \in B \cap V(p)$, so that $B$ is regular. 
\end{proof}

\begin{lemma}\label{lemmTermiAbove}
Let $n\geq 0$ and $\DiamMod n=(W,{\peq},S,V)$. Then, if $x \in W$ is $m$-initial, there is $y \seq x$ which is $m$-terminal.
\end{lemma}

\proof
Let $x = (i,b) $ and $a $ be such that $x \in B= B_m(a,b)$. Since $x$ is $m$-initial, by Lemma \ref{lemmHeightClass}.\ref{itHeightClassInit} we have that $b < n - m + 1$. Hence if we set $b' = n - m +1$ and $B' = B_m (a,b')$, we see by Lemma \ref{lemmHeightClass}.\ref{itHeightClassTerm} that $B'$ is terminal, and as $b \leq b'$ that $y =(i,b')\seq x$, as needed.
\endproof

\begin{lemma}\label{lemmSameHeight}
Let $n\geq 0$ and $m \in [0,n]$.
If $B $ and $B'$ are regular $m$-blocks then $B$ and $B'$ are congruent if and only if they have the same height.
\end{lemma}

\proof
Suppose that $B= B_m(a,b)$ and $B' = B_m(a',b' )$ are regular, so that by Lemma \ref{lemmHeightClass}.\ref{itHeightClassReg}, $b, b' > n - m + 1$.
If $b = b'$ and $k \in [1,2^m)$ then by Lemma \ref{lemInterK},
\[\ell \big ( (a-1) 2^m + k \big ) = \ell(k) = \ell \big ( (a'-1) 2^m + k \big ) ,\]
so that $p \in V \big ( (a-1) 2^m + k, b \big ) $ if and only if $ p \in V \big ( (a'-1) 2^m + k, b \big ) $. Since $b > n - m + 1$ and $\ell ( c 2^m) = \ell( c ) + m > m$ for $c \in \{a,a'\}$, we see that $ b > n- \ell(a2^m)$ and also $b> n - \ell(a' 2^m)$, so that $p \in V (a 2^m, b) \cap V (a' 2^m , b') $. We conclude that for all $k \in [0, 2^m]$, $ p \in \big ( (a-1) 2^m + k, b \big )  $ if and only if $p \in \big ( (a'-1) 2^m + k, b \big ) $, i.e.~$B$ and $B'$ are congruent.

If instead $b \not = b'$, assume without loss of generality that $b < b'$. Since $b > n - m + 1$ we have that $k : = 2^{n-b} \in [1,2^m)$. But then $ b \leq n - (n - b) = n - \ell \big ( (a-1) 2^m + k \big )$, while $ b' > n - (n - b ) = n - \ell \big ((a'-1) 2^m + k \big )$. We conclude that $ p \not \in \big ((a-1) 2^m + k, b \big )  $ while $p \in V\big ((a ' - 1) 2^m + k, b' \big ) $, hence $B$ and $B'$ are not congruent.\david{Double-check.}
\endproof


In order to prove that $(\sim_m)_{m<n}$ is indeed a graded $\R$-bisimulation we will need to consider some basic transformations on blocks.
Namely, we define the {\em successor} of $B_m(a,b)$ to be $B_m([a + 1]_{2^{n-m}},b)$, and if $m=m'+1$, then we say that $B_{m'}(2a - 2 ,b)$ is the {\em first half} of $B$, and $B_{m'}(2a - 1 ,b)$ is the {\em second half} of $B$.

\begin{lemma}\label{lemBlockTrans}
If $B$ and $B'$ are congruent $m$-blocks, then:

\begin{enumerate}

\item the first halves of $B$ and $B'$ are congruent,

\item the second halves of $B$ and $B'$ are congruent, and

\item the successors of the second halves of $B$ and $B'$ are congruent.

\end{enumerate} 
\end{lemma}

\proof
The first two items follow directly from the definition of congruence.
For the third item, the congruence of the successors of the second halves of $B$ and $B'$ is shown by a case-by-case analysis: if $B$ and $B'$ are regular, then the successors of the second halves are either both terminal or both regular with the same height. If $B,B'$ are not regular, then said successors are both initial.
\endproof

\begin{proposition}\label{propRBisim}
The relations $(\sim_m)_{m<n}$ form a graded $\R$-bisimulation on $\DiamMod n$.
\end{proposition}

\proof
Note that $\sim_m$ is symmetric, so we only check the `forth' clauses. Below, assume that $x = (x_1,x_2)$, $y = (y_1,y_2)$ and $x \sim_m y$.\\

\noindent {\sc Atoms:} From $x_1 \equiv y_1 \pmod{2^m}$ and the fact that $x,y$ belong to congruent $m$-blocks, we obtain that $p \in V(x)$ if and only if $p \in V( y )$.
\\

\noindent {\sc Forth $\peq$:} Let $x' = (x_1,b)\seq x$. If $x,y$ are $(m+1)$-regular, it follows from Lemma \ref{lemmSameHeight} that $x_2 = y_2$. Thus $y' \eqdef (y_1,b) \seq y$, and it is not hard to see using Lemma \ref{lemmHeightClass}.\ref{itHeightClassReg} that $x',y'$ are both $m$-regular, so that $x' \sim_m y'$ by Lemma \ref{lemmSameHeight}. If $x'$ is $m$-initial, then it follows that $y$ is $m$-initial and we take $y' = y$. If $x'$ is $m$-terminal, then either $y$ is $m$-terminal and we take $y' = y$, or $y$ is $m$-initial so that by Lemma \ref{lemmTermiAbove} there is some $m$-terminal $y'\seq y$; in either case we have that $x' \sim_m y'$.\\

\noindent {\sc Forth $\tnext$:} If $x,y$ belong to the $(m+1)$-blocks $B,B'$, then $S(x),S(y)$ either both belong to first halves of $B,B'$, to their second halves, or to the successors of their second halves. In any case, it follows from Lemma \ref{lemBlockTrans} that they belong to congruent $m$-blocks, and since addition preserves congruence modulo $2^m$ we obtain $S(x) \sim _m S(y)$.\\

\noindent {\sc Forth $\R$:} Suppose that $x,y$ belong to the $(m+1)$-blocks $B,B'$. Let $x'  = S^r(x)$ and note that $s' = ( [ x_1 +  r] _{2^n} , x_2)$. If $x'$ belongs to the same $(m+1)$-block as $x$, take $y' = (y_1 +  r, y_2)$. Then, it is readily verified that, for each $t\leq r$, $(x_1 + t, x_2) \sim_{m} (y_1 +  t, y_2) $.

Otherwise, let $r'$ be the least such that $ y ' \eqdef ( [ y_1 +  r'] _{2^n} , y_2)$ is {\em not} on the same $(m+1)$-block as $y$ and $y_1 + _{2^n} r' \equiv [ x_1 + r ]_{2^n} \pmod{2^{m}}$. Clearly $r'\leq r$, and thus as before we have that for each $t < r'$, $([x_1 + t]_{2^n} , x_2) \sim_{m} ( [ y_1 + t ] _{2^n} , y_2) $; this is seen by noting that $ [ x_1 + t ] _{2^n} \equiv  [ y_1 + t ]_{2^n} \pmod{2^m}$, and $( [ x_1 + t ] _{2^n} , x_2) $, $ ( [ y_1 +t ]_{2^n} , y_2) $ are both either on the first halves of $B$ and $B'$, or both on the second halves, or both on the successor of the second half; the minimality of $r'$ guarantees that no other case is possible.

If $x'$ is $m$-regular then $x$ must be $(m+1)$-regular, from which it is easy to see that $x,x',y,y'$ all share the same height and hence $y' \sim_{m}x'$. Otherwise, $y'$ is $m$-initial. If $x'$ is $m$-initial define $y''=y'$, and if $x'$ is $m$-terminal, choose $y''\seq y'$ which is $m$-terminal. In either case, $y'' \sim_{m}x'$, as needed.
\endproof

\begin{theorem} \label{thm:undefdiamR}
The formula $\diam p$ is not $\lang_{\R}$-definable over the class of persistent models.
\end{theorem}

\proof
Let $\varphi \in \lang_{\R}$, let $ n = \length \varphi $ and consider the model $\DiamMod {n + 1}$. By Lemma \ref{lemmDiamDisagree}, $\DiamMod {n + 1} ,(1,0)\not \models \diam p$ and $\DiamMod {n + 1} , (1,1) \models \diam p$. However, by Lemma \ref{lemmHeightClass}.\ref{itHeightClassInit}, $B_{n}(1,0) $ and $B_{n} (1,1)$ are both initial, hence $( 1 , 0 ) \sim_{n} ( 1 , 1 )$. By Lemma \ref{lemma:bisimulation:release}, $\DiamMod {n + 1}, ( 1 , 0 ) \models \varphi$ if and only if $\DiamMod {n + 1}, ( 1 , 1 ) \models \varphi$. It follows that $\varphi $ is not equivalent to $\diam p$ over the class of persistent models.
\endproof

\section{Conclusions}

We have studied $\iltl$, an intuitionistic analogue of $\sf LTL$ based on expanding domain models from modal logic and first introduced in \cite{Boudou2017}. In the literature, intuitionistic modal logic is typically interpreted over persistent models, but as we have shown this interpretation has the technical disadvantage of not enjoying the finite model property. Of course, this fact alone does not imply that $\itlb$ is undecidable, and whether the latter is true remains an  open problem.
This should not be surprising, as decidability for intuitionistic modal logics with a transitive modal accessibility relation is notoriously difficult to prove \cite{Simpson94}, having resisted proof techniques that have been successfully applied to other intuitionistic modal logics, such as those in e.g.~\cite{AlechinaS06}.
Meanwhile, our semantics are natural in the sense that we impose the minimal conditions on $\msuccfct$ so that all truth values are monotone under $\peq$, and a wider class of models is convenient as they can more easily be tailored for specific applications. Furthermore, we have presented the notions of bounded bisimulations and shown that, as happens in other modal intuitionistic logics or modal intermediate logics, modal operators are not interdefinable. 

This work and \cite{Boudou2017} represent the first attempts to study $\iltl$. Needless to say, many open questions remain. We know that $\iltl$ is decidable, but the proposed decision procedure is non-elementary. However, there seems to be little reason to assume that this is optimal, raising the following question:

\begin{question}
Are the satisfiability and validity problems for $\iltl$ elementary?
\end{question}

Meanwhile, we saw in Theorems \ref{ThmFMP} and \ref{ThmNonFin} that $\iltl$ has the effective finite model property, while $\itlb$ does not have the finite model property at all. However, it may yet be that $\itlb$ is decidable despite this.

\begin{question}
Is $\itlb$ decidable?
\end{question}

Regarding expressive completeness, it is known that $\sf LTL$ is expressively complete~\cite{Kamp68,Rabinovich12,Gabbay80,Hodkinson94}: $\lang_{\Until}$ is expressively equivalent to monadic first-order logic equipped with a linear order and `next' relation~\cite{Gabbay80}.
Persistent models can be viewed as models of first-order intuitionistic logic, and hence we can ask the same question of $\itlb$.

\begin{question}
Is $\lang_{\ubox{\Until}}$ equally expressive to monadic first-order logic over the class of persistent models?
\end{question}

Finally, a sound and complete axiomatization for $\iltl$ remains to be found. In \cite{DieguezCompleteness} we axiomatize the $\ubox$-free fragment of $\iltl$ and we discuss possible axioms for the full language in \cite{BoudouJELIA}, but treating languages with $\ubox$ seems to be a much more difficult problem.

\begin{question}
Do $\iltl$ or $\itlb$ enjoy natural axiomatizations?
\end{question}

\bibliographystyle{plain}

\newcommand\SortNoop[1]{}


\end{document}